\documentclass[journal,draftcls, onecolumn]{IEEEtran}
\usepackage[dvipdfmx]{graphicx}

\usepackage{amsmath}
\usepackage{amssymb}
\usepackage{color}
\DeclareMathAlphabet{\bm}{OML}{cmm}{b}{it}
\newtheorem{theorem}{Theorem}
\newtheorem{lemma}[theorem]{Lemma}
\newtheorem{definition}[theorem]{Definition}
\newtheorem{corollary}[theorem]{Corollary}
\newtheorem{remark}[theorem]{Remark}
\newtheorem{proposition}[theorem]{Proposition}
\newcommand{\qed}{\hfill \IEEEQED}

\newcommand{\bol}[1]{\mathbf{#1}}
\newcommand{\rom}[1]{\mathrm{#1}}

\newcommand{\argmax}{\mathop{\rm argmax}\limits}
\newcommand{\argmin}{\mathop{\rm argmin}\limits}
\allowdisplaybreaks[2]

\begin{document}

\title{Universal Wyner-Ziv Coding for Distortion Constrained General Side-Information}

\author{Shun~Watanabe~\IEEEmembership{Member,~IEEE}       
\thanks{The first author is with the Department
of Information Science and Intelligent Systems, 
University of Tokushima,
2-1, Minami-josanjima, Tokushima,
770-8506, Japan, 
e-mail:shun-wata@is.tokushima-u.ac.jp.}
and Shigeaki~Kuzuoka~\IEEEmembership{Member,~IEEE}       
\thanks{The second author is with the Department of Computer and Communication Sciences,
Wakayama University, Wakayama, 640-8510, Japan, e-mail:kuzuoka@ieee.org.}

\thanks{Manuscript received ; revised }}

\markboth{Journal of \LaTeX\ Class Files,~Vol.~6, No.~1, January~2007}%
{Shell \MakeLowercase{\textit{et al.}}: Bare Demo of IEEEtran.cls for Journals}

\maketitle
\begin{abstract}
We investigate the Wyner-Ziv coding in which the statistics of
the principal source is known but the statistics of the channel generating
the side-information is unknown except that it is in a certain class.
The class consists of channels such that the distortion between the
principal source and the side-information is smaller than a threshold,
but channels may be neither stationary nor ergodic.
In this situation, we define a new rate-distortion function as the minimum rate such 
that there exists a Wyner-Ziv code that is universal for every channel in the class.
Then, we show an upper bound
and a lower bound on the rate-distortion function, and derive a matching condition
such that the upper and lower bounds coincide. 
The relation between the new rate-distortion function and the rate-distortion
function of the Heegard-Berger problem is also discussed.
%\boldmath
\end{abstract}


\IEEEpeerreviewmaketitle

\section{Introduction}

In the seminal paper \cite{wyner:76}, Wyner and Ziv characterized 
the rate-distortion function of the lossy source coding
with side-information at the decoder (See Fig. \ref{Fig1}). 
In this paper, we consider a universal coding of this problem
where  the statistics of the principal source is known but the channel from
the principal source to the side-information is unknown except that it is 
in a certain class.

To motivate the problem setting investigated in this paper, 
let us consider the following practical situation first.
Suppose that the decoder already has a lossy compressed version of 
the principal source, and want to get a refined one.
The encoder does not know how the previously transmitted
lossy version is encoded, but knows that the quality of
the lossy version is guaranteed to be above a certain level.
What is the minimum additional rate that must be transmitted by the encoder
so that the quality of the refined version is above a required level?

The above mentioned situation can be modeled as follows.
The principal source $X^n$ is a known i.i.d. source, and 
the side-information $Y^n$ is generated from $X^n$ through 
a channel $W^n$. The statistical property of the channel is unknown, but
the distortion caused by the channel is smaller than a certain level $E$
for a prescribed distortion measure. We assume that the distortion
measure is additive, but the channel may be neither stationary nor ergodic.
We consider the maximum distortion constraint and the average distortion
constraint for the channel. Since we allow non-ergodic channel,
the class of channels constrained by the maximum distortion and that constrained by
the average distortion are different. 
In this problem formulation, we are interested in the minimum rate
$R_m(D|E)$ and $R_a(D|E)$ such that the reproduction with distortion level $D$ is possible
at the decoder for any channel in the classes of channels satisfying the distortion
level $E$ with the maximum distortion constraint and
the average distortion constrain respectively.
In other word, we are interested in the minimum rate such that the universal coding is possible for
each class. 

For the maximum distortion constrained class, we show an upper bound and a lower bound on
$R_m(D|E)$. We also derive a matching condition such that the upper 
and the lower bounds coincide. Especially, for the binary Hamming
example, we show that the matching condition is satisfied, and thus 
$R_m(D|E)$ is completely characterized.

For the average distortion constrained class, we show an upper bound and a lower bound  
on $R_a(D|E)$. 
For the case with $D=0$, i.e., the loss less reproduction case,
we show that the upper and lower bounds coincide and thus $R_a(0|E)$ is completely characterized.
Surprisingly, $R_a(0|E) = H(X)$, i.e., the side-information is completely useless, 
for any $E > 0$. 

Some remarks on related literatures are in order.

For lossless source coding with side-information, i.e.,
the Slepian-Wolf network \cite{slepian:73},  the existence of universal code was first shown by
Csisz\'ar and K\"orner \cite{csiszar:81} (existence of
linear universal code was also shown by Csisz\'ar \cite{csiszar:82}).
After that, the universal codings for the Slepian-Wolf network
or other related lossless multi-terminal networks were studied by several
researchers \cite{oohama:94, oohama:96, kimura:09}.

For lossy source coding with side-information, i.e.,
the Wyner-Ziv network, the universal coding problem was 
investigated by Merhav and Ziv \cite{merhav:06}, Jalali {\em et.~al.}~\cite{jalali:10},
and  Reani and Merhav \cite{reani:11}. It should be noted that the universal
codes proposed in these literatures are universal for the statistics of the principal
source but not for the channel generating the side-information, i.e., the statistics
of the channel is known at the encoder. Under the same condition, i.e., known channel,
it is also known that the universal code can be constructed for the network with
several decoders \cite{kuzuoka:10b}.

%The encoder of the
%Wyner-Ziv coding typically consists of two phases, the quantization and
%the bin coding, and the quantization needs to be adapted to the channel
%generating the side-information. For this reason, there has not been any known 
%Wyner-Ziv coding that is universal for the channel statistics, and it is
%considered as a challenging problem.
%To circumvent this difficulty, we construct our Wyner-Ziv code that is adapted
%to the worst channel in a certain sense by using the saddle point argument.
%Although existence of the worst channel is not always guaranteed, which causes 
%a gap between the upper and the lower bounds,
%for the binary Hamming example we show that the worst channel exists 
%and it is the binary symmetric channel.

The universal Wyner-Ziv coding is also related to the Heeger-Berger
problem \cite{heegard:85}, in which there are several decoders that
have their own side-information. The Heeger-Berger problem has not
been solved in general, and it has only been solved under the condition   
that there is a degraded partial order between the channels generating the side-information
\cite{steinberg:04, tian:07, tian:08} except some special cases \cite{timo:11b, watanabe:11}. 
It should be noted that there is no degraded partial order among the channel class 
considered in this paper. Thus, the authors believe that the result in this paper also
shed some light on the unsolved Heeger-Berger problem.

Our problem setting can be also viewed as a kind of 
the successive refinement coding \cite{equitz:91, rimoldi:94}.
The successive refinement coding consists of two layers
of the encodings. If the method used by the first layer encoder
is not known to the second layer encoder, this is exactly
the situation of our problem setting.

Although the universal coding for distortion constrained 
class of channels is unfamiliar and new in the source coding
scenario, this kind of channel is quite natural when the channel is cased by an
adversary such as in the data hiding scenario.
Indeed, this kind of channel class is commonly used in 
the information theoretical analysis of the
data hiding \cite{moulin:03, cohen:02, baruch:03}. 
%It should be also noted that the data hiding is a kind of
%the channel coding with non-causal state information known 
%at the encoder and decoder.
%Furthermore, it is known that there is a duality between the channel coding
%with non-causal state information at the encoder and decoder
%and the lossy source coding
%with side-information at the decoder \cite{barron:03, cover:02}.
%For this reason, our problem setting can be regarded as 
%a dual problem of the data hiding scenario investigated in 
%the literatures \cite{moulin:03, cohen:02, baruch:03}. 

There are some technical differences between 
the data hiding problem and our problem. 
First, in the data hiding problem, the channel output is only
used for the decoding of the encoded message. On the other hand,
in our problem, the side-information is not only used for the decoding
of the encoded source, but also for the estimation at the decoder.
This makes the problem difficult, and causes a gap between the
upper bound and the lower bound derived in this paper. 
Second, in the data hiding problem for the average distortion
constrained class of channels, it was shown that the achievable
transmission rate is $0$, i.e., the channel is completely useless \cite{cohen:02}. 
On the other hand, in our problem for
the average distortion constrained class of channels, the 
side-information is useless for bin coding, but it can be used
for the estimation at the decoder. Thus, $R_a(D|E)$ can be
strictly smaller than the rate-distortion function $R(D)$ without
any side-information for $D > 0$, though $R_a(0|E) = H(X)$.

The rest of this paper is organized as follows.
In Section \ref{section:preliminaries}, we introduce notations and the formal definition of the
problem. In Section \ref{section:main}, we state our main theorems, and
show a representative example, i.e., the binary Hamming example.
In Sections \ref{section:proof-of-theorem:main} and \ref{section:proof-of-theorem:main-2}, we present proofs of
the main theorems.

%%%%%%%% Fig %%%%%%%%%%%%%%%
\begin{figure}[t]
\centering
\includegraphics[width=0.7\linewidth]{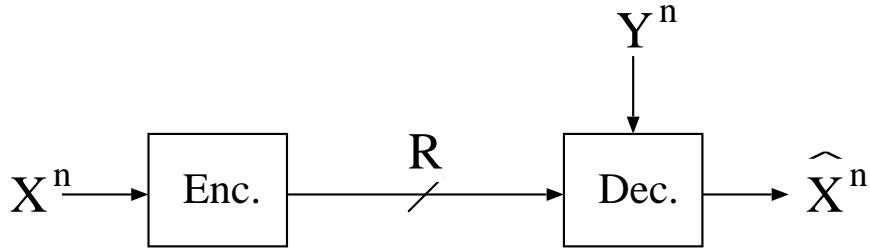}
\caption{The Wyner-Ziv coding system.}
\label{Fig1}
\end{figure}

%%%%%%% Problem Formulation %%%%%%

\section{Preliminaries}
\label{section:preliminaries}

\subsection{Notations}

Henceforth, we adopt the following notation conventions.
Random variables will be denoted by capital letters such as $X$,
while their realizations will be denoted by respective lower case letters such as $x$.
A random vector of length $n$ is denoted by $X^n = (X_1,\ldots,X_n)$, 
while its realization is denoted by $x^n = (x_1,\ldots,x_n)$.
The alphabet of a random variable is denoted by a calligraphic letter such as ${\cal X}$,
and its $n$-fold Cartesian product is denoted by ${\cal X}^n$.
The probability distribution of random variable $X$ is denoted by $P_X$,
and its $n$-fold i.i.d.~extension is denoted by $P_X^n$. 
For a given channel $W$, its $n$-fold i.i.d.~extension is denoted by $W^{\times n}$,
while $W^n$ indicates a channel that is not necessarily i.i.d.. 
The set of all probability distribution
on ${\cal X}$ is denoted by ${\cal P}({\cal X})$. The set of all channel from 
${\cal X}$ to ${\cal Y}$ is denoted by ${\cal P}({\cal Y}|{\cal X})$. 
The indicator function is denoted by $\bol{1}[\cdot]$.
The entropy and the mutual information is denoted in
a standard notation such as $H(X)$ or $I(X;Y)$.
For a input distribution $P$ of a channel $W$, we sometimes use
the notation $I(P,W)$ to designate the mutual information $I(X;Y)$,
where the joint distribution of $(X,Y)$ is $P(x)W(y|x)$.
The variational distance between two distributions $P$ and $Q$ is 
denoted by $\| P - Q \|$.
In the proofs of our main theorems, we extensively use the 
type and typicality, which are summarized in Appendix \ref{Appendix:type}.

\subsection{Problem Formulation}
\label{section:problem}

Let $\bm{X} = \{ X^n \}_{n=1}^\infty$ be an i.i.d. source.
Let 
\begin{eqnarray*}
e_n(x^n, y^n) := \frac{1}{n} \sum_{t=1}^n e(x_t, y_t)
\end{eqnarray*}
be an additive distortion measure for side information.
As a natural assumption, we assume that there exists $y$
such that $e(x,y) = 0$ for each $x$.
We also assume that the distortion is bounded, i.e., 
$e(x,y) \le e_{\max} < \infty$ for every $(x,y)$.
For a given distortion $E \ge 0$, 
we consider the following maximum distortion constraint on the side-information
\begin{eqnarray}
{\cal W}_m(E) &:=& \left\{ \bm{W} = \{ W^n \}_{n=1}^\infty : \forall \delta > 0~\exists n_0(\delta)~\mbox{s.t.}~  \right.  \nonumber \\
	&& \left. \Pr\{ e_n(X^n,Y^n) > E \} \le \delta ~\forall n \ge n_0(\delta) \right\}, 
\end{eqnarray}
where $Y^n$ is the output of channel $W^n$ with input $X^n$.
It should be noted that $n_0(\delta)$ depends on $\delta$ but not on $\bm{W}$.
We also consider the average distortion constraint
\begin{eqnarray}
\lefteqn {{\cal W}_a(E) } \nonumber \\
&:=& \left\{ \bm{W} = \{ W^n \}_{n=1}^\infty : e_n(P_{X^n},W^n ) \le E ~\forall n \ge 1 \right\} 
\end{eqnarray}
 where
\begin{eqnarray*}
e_n(P_{X^n},W^n) &:=& \mathbb{E}[e_n(X^n,Y^n)] \\
 &=& \sum_{x^n,y^n} P_X^n(x^n) W^n(y^n|x^n) e_n(x^n,y^n).
\end{eqnarray*}
As it will be clarified later, the maximum distortion constraint and the average distortion constraint
are completely different.

Let $\hat{{\cal X}}$ be the reproduction alphabet.
Then, let 
\begin{eqnarray}
\label{eq:additive-distortion-measure}
d_n(x^n, \hat{x}^n) := \frac{1}{n} \sum_{t=1}^n d(x_t, \hat{x}_t)
\end{eqnarray}
be an additive distortion measure for reproduction.
We assume 
$d(x,\hat{x}) \le d_{\max} < \infty$ for every $(x,\hat{x})$.
We also assume that for each $x$ there exists $\hat{x}$ such that $d(x,\hat{x}) = 0$.

We consider (possibly stochastic) encoder
\begin{eqnarray*}
\varphi_n:{\cal X}^n \to {\cal M}_n
\end{eqnarray*}
and decoder
\begin{eqnarray*}
\psi_n:{\cal M}_n \times {\cal Y}^n \to \hat{{\cal X}}^n.
\end{eqnarray*}

\begin{definition}
For any $\varepsilon > 0$, if there exists $n_0(\varepsilon)$ 
and a sequence of codes $\{(\varphi_n,\psi_n) \}_{n=1}^\infty$ such that
\begin{eqnarray}
\label{eq:rate-requirement}
\frac{1}{n} \log |{\cal M}_n| \le R + \varepsilon
\end{eqnarray}
and
\begin{eqnarray}
\label{eq:distortion-requirement}
\mathbb{E}[d_n(X^n,\psi_n(\varphi_n(X^n),Y^n))]  \le D + \varepsilon
\end{eqnarray}
for every $\bm{W} \in {\cal W}_m(E)$ and $n \ge n_0(\varepsilon)$, then
we define the rate $R$ to be {\em achievable}.
We also define the rate distortion function
\begin{eqnarray*}
R_m(D|E) := \inf\{ R : R \mbox{ is achievable} \}.
\end{eqnarray*}
We also define $R_a(D|E)$ by replacing ${\cal W}_m(E)$ with ${\cal W}_a(E)$.
\end{definition}

From the problem formulation, we can prove the following relation between $R_m(D|E)$ and $R_a(D|E)$,
which will be proved in Appendix \ref{appendix:relation-max-ave}.
\begin{proposition} \label{proposition:relation-max-ave}
We have
\begin{eqnarray*}
R_a(D|E) \ge \lim_{\epsilon \downarrow 0} R_m(D|E-\epsilon).
\end{eqnarray*}
\end{proposition}

\begin{remark}
As we can find from the proof of Theorem \ref{theorem:main},
the theorem holds even if
the average distortion requirement in (\ref{eq:distortion-requirement}) is replaced 
by the maximum distortion requirement
\begin{eqnarray}
\label{eq:distortion-requirement-2}
\Pr\left\{ d_n(X^n, \psi_n(\phi_n(X^n),Y^n)) > D + \varepsilon \right\} \le \varepsilon.
\end{eqnarray}
However, Theorem \ref{theorem:main-2} does not hold if 
(\ref{eq:distortion-requirement}) is replaced by (\ref{eq:distortion-requirement-2}).
\end{remark}

Let 
$R_{WZ}(D|W)$
be the rate distortion function of the ordinary Wyner-Ziv problem 
in which the principal source is $X$ and the side-information $Y$
is the output of the channel $W \in {\cal P}({\cal Y}|{\cal X})$. 

The rate distortion function $R_m(D|E)$ (or $R_a(D|E)$) means that 
if $R > R_m(D|E)$ there exists a {\em universal} code that works well for 
every $\bm{W} \in {\cal W}_m(E)$ (or $\bm{W} \in {\cal W}_a(E)$).
It should be noted that this definition of universality is different from
the ordinary definition of the universality.
Let 
\begin{eqnarray}
{\cal W}_{WZ}(R,D) := \left\{ W \in {\cal P}({\cal Y}|{\cal X}) : R_{WZ}(D|W) \le R \right\}.
\label{eq:class-ordinary-sense}
\end{eqnarray}
In the ordinary definition of the universality, we require that there exists a code 
that works well for every $W \in {\cal W}_{WZ}(R,D)$.
This requirement seems much more severe than the requirement of $R_m(D|E)$ (or $R_a(D|E)$),
which will be discussed in more detail in Section \ref{section:main}.

%%%%% Heegard Berger Problem %%%%%%%
\subsection{Heegard-Berger Problem}
\label{section:heegard-berger-problem}

For later use, we review the problem formulation of the
Heegard-Berger (HB) problem \cite{heegard:85} in this section.
We restrict our attention to the case with two decoders (see Fig.~\ref{Fig2}).
Furthermore, we restrict our attention to the case such that 
the alphabets of the side-information,
the reproduction alphabets, and the distortion measures 
for both the decoders are common, which are  denoted by ${\cal Y}$,
$\hat{{\cal X}}$, and $d(\cdot,\cdot)$ respectively.

Let us consider the HB coding for i.i.d.~joint source $(X,Y_1,Y_2)$.
The HB code consists of one encoder
\begin{eqnarray*}
\varphi_n^{HB}:{\cal X}^n \to {\cal M}_n
\end{eqnarray*}
and two decoders
\begin{eqnarray*}
&& \psi_n^{HB1}:{\cal M}_n \times {\cal Y}^n \to \hat{{\cal X}}^n, \\
&& \psi_n^{HB2}:{\cal M}_n \times {\cal Y}^n \to \hat{{\cal X}}^n.
\end{eqnarray*}
For a pair $(D_1,D_2)$ of distortions, a rate $R$ is defined to be $(D_1,D_2)$-achievable
if, for any $\varepsilon > 0$, there exists a sequence of HB code $\{ (\varphi_n^{HB}, \psi_n^{HB1}, \psi_n^{HB2}) \}_{n=1}^{\infty}$
such that 
\begin{eqnarray*}
\frac{1}{n} \log |{\cal M}_n| &\le& R + \varepsilon, \\
\mathbb{E}\left[ d_n(X^n,\hat{X}_i^n) \right] &\le& D_i + \varepsilon~~~i=1,2, 
\end{eqnarray*}
for sufficiently large $n$, where $\hat{X}_i^n = \psi_n^{HBi}(\varphi_n(X^n), Y_i^n)$
and $d_n$ is defined in (\ref{eq:additive-distortion-measure}).
Then, the HB rate-distortion function $R_{HB}(D_1,D_2|X,Y_1,Y_2)$ for $(X,Y_1,Y_2)$
is defined as the infimum of $(D_1,D_2)$-achievable rate $R$.

Fix an i.i.d.~source $P_X$. Then two side-information channel $W_1:{\cal X} \to {\cal Y}$
and $W_2:{\cal X} \to {\cal Y}$ define an i.i.d.~joint source $(X,Y_1,Y_2)$ whose joint 
distribution $P_{XY_1Y_2}$ is given by $P_{XY_1Y_2}(x,y_1,y_2) = P_X(x) W_1(y_1|x)W_2(y_2|x)$,
where $x \in {\cal X}$ and $y_1,y_2 \in {\cal Y}$. In the following, we denote by
$R_{HB}(D_1,D_2|W_1,W_2)$ the HB rate-distortion function $R_{HB}(D_1,D_2|X,Y_1,Y_2)$ for 
$(X,Y_1,Y_2)$ defined by $W_1$ and $W_2$.

Unfortunately, finding a single-letter expression for $R_{HB}(D_1,D_2|W_1,W_2)$ has been a 
long-standing open problem. So, we consider a special case. Let
\begin{eqnarray*}
E_* := \min_{y \in {\cal Y}} \sum_{x \in {\cal X}} P_X(x) e(x,y)
\end{eqnarray*}
and $y_* \in {\cal Y}$ be a symbol which attains the minimum. Further, let $W_*:{\cal X} \to {\cal Y}$ be
a side-information channel such that $W_*(y_*|x) = 1$ irrespective $x \in {\cal X}$. Then, let us consider
a special case where $W_1 = W_*$. This case is equivalent to the problem of "lossy coding when side-information
may be absent". Heegard and Berger \cite{heegard:85} (see also \cite{elgamal-kim-book}) showed the following.
\begin{proposition}[\cite{heegard:85}]
\label{proposition:may-be-absent}
We have
\begin{eqnarray*}
R_{HB}(D_1,D_2|W_*,W_2) = \min\left[ I(X;\hat{X}_1) + I(X;V|\hat{X}_1,Y) \right]
\end{eqnarray*}
where $\min$ is taken over all conditional distribution $P_{V\hat{X}_1|X}$ with 
$|{\cal V}| \le |{\cal X} \times \hat{{\cal X}}| + 2$ and functions $f:{\cal V} \times \hat{{\cal X}} \times {\cal Y} \to \hat{{\cal X}}$
such that 
\begin{eqnarray*}
\mathbb{E}[d(X,X_1) ] &\le& D_1, \\
\mathbb{E}[d(X,f(V,X_1,Y))] &\le& D_2.
\end{eqnarray*}
\end{proposition}

%%%%%%%% Fig %%%%%%%%%%%%%%%
\begin{figure}[t]
\centering
\includegraphics[width=0.7\linewidth]{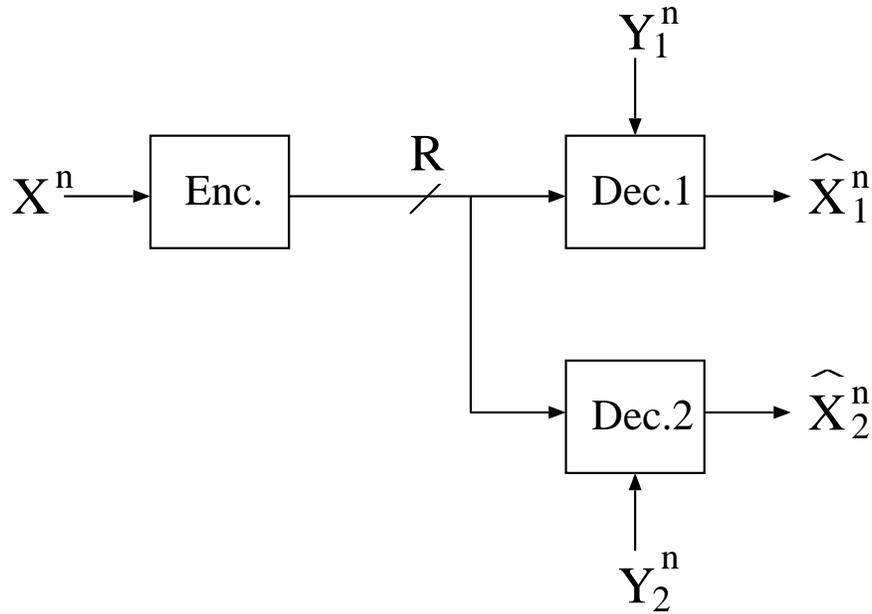}
\caption{The Heegard-Berger coding system.}
\label{Fig2}
\end{figure}

%%%%% Main Results %%%%%%%%%%%%%
\section{Main Result}
\label{section:main}

%%% Preliminaries %%%%%%
\subsection{Convex Form of WZ Rate-Distortion Function}

We need convex form of the Wyner-Ziv rate-distortion function
introduced in \cite{willems:83}.
Let ${\cal U}$ be the set of all functions from ${\cal Y}$ to $\hat{{\cal X}}$.
The set ${\cal U}$ includes a constant function, i.e., $u(y) = \hat{x}~\forall y \in {\cal Y}$
for each $\hat{x} \in \hat{{\cal X}}$. We denote the set of constant functions by 
$\bar{{\cal U}} \subset {\cal U}$.
For fixed channel $W \in {\cal P}({\cal Y}|{\cal X})$ and fixed test channel $V \in {\cal P}({\cal U}|{\cal X})$, we denote
\begin{eqnarray*}
d(V,W) := \sum_{u,x,y} P_X(x) V(u|x)W(y|x) d(x,u(y)).
\end{eqnarray*}
For a fixed channel $W \in {\cal P}({\cal Y}|{\cal X})$, let
\begin{eqnarray*}
{\cal V}(W,D) := \left\{ V \in {\cal P}({\cal U}|{\cal X}) : d(V,W) \le D \right\}. 
\end{eqnarray*}
Let 
\begin{eqnarray*}
{\cal W}_1(P_X,E) &=& {\cal W}_1(E) \\
&:=& \{ W  \in {\cal P}({\cal Y}|{\cal X}) : e(P_X,W) \le E \}
\end{eqnarray*}
and
\begin{eqnarray*}
{\cal V}(E,D) := \left\{ V \in {\cal P}({\cal U}|{\cal X}) : d(V,W) \le D~\forall W \in {\cal W}_1(E) \right\}.
\end{eqnarray*}
For $(V,W) \in {\cal P}({\cal U}|{\cal X}) \times {\cal P}({\cal Y}|{\cal X})$, let
\begin{eqnarray}
\phi(V,W) &:=& I(U;X) - I(U;Y) \label{eq:f-1} \\
	&=& I(U; X|Y).
	\label{eq:f-2}
\end{eqnarray}
Note that $\phi(\cdot, W)$ is a convex function for fixed $W$,
which can be confirmed from (\ref{eq:f-2}), and $\phi(V,\cdot)$
is a concave function for fixed $V$, which can be confirmed from (\ref{eq:f-1}).

By the above notations, the Wyner-Ziv rate-distortion function is given by
\begin{eqnarray*}
R_{WZ}(D|W) = \min_{V \in {\cal V}(W,D)} \phi(V,W).
\end{eqnarray*}

%We have the following properties.
%\begin{lemma}
%\label{lemma:concavity-r-d}
%The Wyner-Ziv rate-distortion function $R_{WZ}(D|W)$ is a concave 
%function of the channel, i.e., 
%\begin{eqnarray*}
%\lefteqn{ R_{WZ}(D|\lambda W_1 + (1-\lambda) W_2) } \\
%	&\ge& \lambda R_{WZ}(D|W_1) + (1-\lambda) R_{WZ}(D|W_2)
%\end{eqnarray*}
%holds for $W_1,W_2 \in {\cal P}({\cal Y}|{\cal X})$
%and $0 \le \lambda \le 1$.
%\end{lemma}

Let 
\begin{eqnarray*}
\tilde{R}_{WZ}(D|W,E) = \min_{V \in {\cal V}(E,D)} \phi(V,W)
\end{eqnarray*}
be the pseudo rate-distortion function.
\begin{lemma}
\label{lemma:cocavity-pseudo}
The pseudo rate-distortion function $\tilde{R}_{WZ}(D|W,E)$ is a concave 
function of the channel, i.e., 
\begin{eqnarray*}
\lefteqn{ \tilde{R}_{WZ}(D|\lambda W_1 + (1-\lambda) W_2,E) } \\
	&\ge& \lambda \tilde{R}_{WZ}(D|W_1,E) + (1-\lambda) \tilde{R}_{WZ}(D|W_2,E)
\end{eqnarray*}
holds for $W_1,W_2 \in {\cal P}({\cal Y}|{\cal X})$
and $0 \le \lambda \le 1$.
\end{lemma}
\begin{proof}
Let 
\begin{eqnarray*}
\hat{V} = \argmin_{ V \in {\cal V}(E, D)} \phi(V, \lambda W_1 + (1-\lambda) W_2).
\end{eqnarray*}
Then, we have
\begin{eqnarray*}
\lefteqn{ R_{WZ}(D| \lambda W_1 + (1-\lambda) W_2,E) } \\
&=& \phi(\hat{V}, \lambda W_1 + (1-\lambda) W_2) \\
&\ge& \lambda \phi(\hat{V}, W_1) + (1- \lambda) \phi(\hat{V}, W_2) \\
&\ge& \lambda \min_{V \in {\cal V}(E,D)} \phi(V,W_1) + (1-\lambda) \min_{V \in {\cal V}(E,D)} \phi(V,W_2) \\
&=& \lambda R_{WZ}(D|W_1) + (1-\lambda) R_{WZ}(D|W_2),
\end{eqnarray*}
where we used concavity of $\phi(V,\cdot)$ for fixed $V$ in the first inequality.
\end{proof}

%%%% Statement of Results %%%%
\subsection{Statements of General Results}

For the maximum distortion class, we have the following.
\begin{theorem}
\label{theorem:main}
We have
\begin{eqnarray}
R_m(D|E) &\ge& \max_{W \in {\cal W}_1(E)} R_{WZ}(D|W) \label{eq:converse-1} \\
	&=& \max_{W \in {\cal W}_1(E)} \min_{V \in {\cal V}(W,D)} \phi(V,W)
	\label{eq:converse-2}
\end{eqnarray}
and
\begin{eqnarray}
R_m(D|E) &\le& \min_{V \in {\cal V}(E,D)} \max_{W \in {\cal W}_1(E)} \phi(V,W) \label{eq:direct-1} \\
	&=& \max_{W \in {\cal W}_1(E)} \min_{V \in {\cal V}(E,D)} \phi(V,W). \label{eq:direct-2} \\
	&=& \max_{W \in {\cal W}_1(E)} \tilde{R}_{WZ}(D|W,E)
\end{eqnarray}
\end{theorem}
\begin{proof}
See Section \ref{section:proof-of-theorem:main}.
\end{proof}

\begin{remark}
Technically in the converse part, at $D=0$, we only have the inequality
\begin{eqnarray*}
R_m(0|E) \ge \lim_{\epsilon \downarrow 0} \max_{W \in {\cal W}_1(E)} R_{WZ}(\epsilon|W).
\end{eqnarray*} 
This is because we use the fact that $R_{WZ}(D|W)$ is a continuos function with respect to $W$ in the converse proof
(see Section \ref{subset:converse-proof-maximum}),
and it is not clear whether $R_{WZ}(D|W)$ is a continuos function with respect to $W$ at $D=0$ in general.
\end{remark}

The difference between (\ref{eq:converse-2}) and (\ref{eq:direct-2}) are
${\cal V}(W,D)$ and ${\cal V}(E,D)$. Thus, we have the following matching conditions.

\begin{corollary}
\label{corollary:matching-1}
Let $(V^*,W^*)$ be a saddle point satisfying
\begin{eqnarray*}
\phi(V^*,W^*) = \max_{W \in {\cal W}_1(E)} \min_{V \in {\cal V}(E,D)} \phi(V,W).
\end{eqnarray*}
Suppose that 
\begin{eqnarray*}
\hat{V} := \argmin_{V \in {\cal V}(W^*,D)} \phi(V,W^*) \in {\cal V}(E,D).
\end{eqnarray*}
Then, we have
\begin{eqnarray*}
R_m(D|E) = \phi(V^*,W^*) = \max_{W \in {\cal W}_1(E)} \min_{V \in {\cal V}(E,D)} \phi(V,W).
\end{eqnarray*}
\end{corollary}
\begin{proof}
We have
\begin{eqnarray*}
\phi(\hat{V}, W^*) &=& \min_{V \in {\cal V}(W^*,D)} \phi(V, W^*) \\
	&\le& \max_{W \in {\cal W}_1(E )} \min_{V \in {\cal V}(W, D)} \phi(V,W) \\
	&\le& R_m(D|E) \\
	&\le& \phi(V^*, W^*) \\
	&=& \min_{V \in {\cal V}(E,D)} \phi(V,W^*) \\
	&\le& \phi(\hat{V}, W^*).
\end{eqnarray*}
\end{proof}

\begin{corollary}
\label{corollary:matching-2}
Under the same notations as Corollary \ref{corollary:matching-1}, suppose that
\begin{eqnarray}
\label{eq:support-condition}
\mbox{supp}(\hat{V}) \subset \bar{{\cal U}}.
\end{eqnarray}
Then, we have
\begin{eqnarray*}
R_m(D|E) = \phi(V^*,W^*) = \max_{W \in {\cal W}_1(E)} \min_{V \in {\cal V}(E,D)} \phi(V,W).
\end{eqnarray*}
\end{corollary}
\begin{proof}
When (\ref{eq:support-condition}) is satisfied, the distortion
\begin{eqnarray*}
\sum_{u,x,y} P_X(x) \hat{V}(u|x) W(y|x) d(x,u(y))
\end{eqnarray*}
does not depend on the channel $W$. Thus, Corollary \ref{corollary:matching-1} implies
the statement of the present corollary 
\end{proof}

For the average distortion class, we have the following.
\begin{theorem}
\label{theorem:main-2}
We have
\begin{eqnarray}
\label{eq:rd-average-class-converse}
R_a(D|E) \ge \max_{\lambda, E_1, E_2, W_1, W_2: \atop \stackrel{\lambda E_1 + (1-\lambda) E_2 \le E}{W_1 \in {\cal W}_1(E_1), W_2 \in {\cal W}_1(E_2)}}
 \min_{D_1, D_2: \atop \lambda D_1 + (1-\lambda) D_2 \le D} R_{HB}(D_1,D_2|W_1,W_2),
\end{eqnarray}
where (i) $\max$ is taken over all $0 \le \lambda \le 1$, $E_j \ge 0$, and side information channels $W_1$, $W_2$ such that
$\lambda E_1 + (1-\lambda) E_2 \le E$ and $W_j \in {\cal W}_1(E_j)$ ($j = 1,2$) and (ii) $\min$ is taken over
all $D_1, D_2 \in [0,d_{\max}]$ such that $\lambda D_1 + (1-\lambda) D_2 \le D$. Especially,
\begin{eqnarray}
\label{eq:rd-average-class-converse-2}
R_a(D|E) \ge \max_{\lambda, E_2, W_2 \in {\cal W}_1(E_2) \atop \lambda E_* + (1-\lambda) E_2 \le E}
 \min_{D_1,D_2: \atop \lambda D_1 + (1-\lambda) D_2 \le D} R_{HB}(D_1,D_2|W_*,W_2)
\end{eqnarray}
holds.
We also have
\begin{eqnarray}
\label{eq:rd-average-class}
R_a(D|E) \le \min_{V \in {\cal V}(E,D)} I(P_X,V).
\end{eqnarray}
\end{theorem}
\begin{proof}
See Section \ref{section:proof-of-theorem:main-2}.
\end{proof}
%%%

\begin{remark}
Note that (\ref{eq:rd-average-class-converse-2}) is obtained from (\ref{eq:rd-average-class-converse}) by
letting $E_1 = E_*$ and $W_1 = W_*$. Thus, (\ref{eq:rd-average-class-converse}) is tighter than
(\ref{eq:rd-average-class-converse-2}). However, we cannot give a single letter expression for the right
hand side of (\ref{eq:rd-average-class-converse}), while we can for (\ref{eq:rd-average-class-converse-2})
by using Proposition \ref{proposition:may-be-absent}.
\end{remark}
\begin{remark}
A close inspection of the proof reveals that we can generalize (\ref{eq:rd-average-class-converse}) by
considering one-to-$m$ lossy source coding with side information at the decoders. That is, in the same
manner as (\ref{eq:rd-average-class-converse}), we can show that 
\begin{eqnarray}
\label{eq:rd-average-class-converse-3}
R_a(D|E) \ge \max_{\vec{\lambda}, \bm{E}, \bm{W}} \min_{\bm{D}} R_{HB}(D_1,D_2,\ldots,D_m|W_1,W_2,\ldots,W_m),
\end{eqnarray}
where (i) $\max$ is taken over all $\vec{\lambda} = (\lambda_1,\ldots,\lambda_m)$,
$\bm{E} = (E_1,\ldots,E_m)$, and $\bm{W} = (W_1,\ldots,W_m)$ such that
$\sum_j \lambda_j = 1$, $\sum_j \lambda_j E_j \le E$, and $W_j \in {\cal W}_1(E_j)$ ($j=1,\ldots,m$) and
(ii) $\min$ is taken over all $\bm{D} = (D_1,\ldots,D_m)$ such that $\sum_j \lambda_j D_j \le D$.
The authors conjecture that the bound (\ref{eq:rd-average-class-converse-3}) is not tighter than
(\ref{eq:rd-average-class-converse}), i.e., is equivalent to (\ref{eq:rd-average-class-converse}).
\end{remark}
\begin{remark}
The upper bound in 
(\ref{eq:rd-average-class}) is derived by using the side-information
only for the estimation at the decoder and not for the bin coding, which
is the difference between (\ref{eq:direct-1}) and (\ref{eq:rd-average-class}).
\end{remark}

From Theorem \ref{theorem:main-2}, we have several corollaries.
At first, let us set parameters in (\ref{eq:rd-average-class-converse-2}) as $\lambda = 0$
and $E_2 = E$. Then, we have 
\begin{eqnarray*}
R_a(D|E) &\ge& \max_{W_2 \in {\cal W}_1(E)} \min_{D_2 \le D} R_{HB}(D_1,D_2|W_*,W_2) \\
&=& \max_{W_2 \in {\cal W}_1(E)} R_{HB}(d_{\max}, D|W_*,W_2).
\end{eqnarray*}
Note that $R_{HB}(d_{\max},D|W_*,W_2)$ equals to the Wyner-Ziv rate-distortion function $R_{WZ}(D|W_2)$.
This fact gives the following corollary.
\begin{corollary}
\label{corollary:average-class-wyner-ziv-bound}
We have
\begin{eqnarray*}
R_a(D|E) \ge \max_{W \in {\cal W}_1(E)} R_{WZ}(D|W).
\end{eqnarray*}
\end{corollary}

Next, let us consider the lossless case, 
i.e., $d(\cdot,\cdot)$ is the Hamming distortion measure and $D=0$.
Note that $R_{HB}(0,0|W_*,W_2)$ equals to the 
minimum coding rate such that the decoder $\psi_n^{HB1}$ without side information can
reproduce $X^n$ in losslessly. Thus, for any side information channel $W_2$, 
\begin{eqnarray*}
R_{HB}(0,0|W_*, W_2) = H(X).
\end{eqnarray*}
Since $R_a(0|E) \le H(X)$, we have the following corollary.
%%%
\begin{corollary}
\label{corollary-average-lossless}
For $D = 0$ and $E > 0$, we have\footnote{We need the condition $E > 0$ because we need
to take $\lambda > 0$ in (\ref{eq:rd-average-class-converse-2}). When $e(\cdot,\cdot)$ is the
Hamming distortion measure and $E=0$, then we have $R_a(0|0) = 0$. However, $R_a(0|0)$ may be
positive in general. For example, $R_a(0|0)$ can be positive for a distortion measure such that $e(x,y)=0$ for every $(x,y)$.}
\begin{eqnarray}
\label{eq:rd-average-class-loss-less}
R_a(0|E) = H(X).
\end{eqnarray}
\end{corollary}
This corollary indicates that the side information is completely useless when $D = 0$ and $E > 0$.
It should be emphasized that Corollary \ref{corollary:average-class-wyner-ziv-bound} does not give 
Corollary \ref{corollary-average-lossless} in general. This means that our result (\ref{eq:rd-average-class-converse-2})
is tighter than Corollary \ref{corollary:average-class-wyner-ziv-bound}.

Lastly, we show that  our bound (\ref{eq:rd-average-class-converse-2}) gives another trivial bound.
Assume that $E \ge E_*$. Then, we can set $\lambda = 1$ in (\ref{eq:rd-average-class-converse-2}) and have
\begin{eqnarray*}
R_a(D|E) &\ge& \max_{W_2 \in {\cal W}_1(E_2)} \min_{D_1 \le D} R_{HB}(D_1,D_2|W_*,W_2) \\
&=& \max_{W_2 \in {\cal W}_1(E_2)} R_{HB}(D,d_{\max}|W_*,W_2).
\end{eqnarray*}
Furthermore, for any side information channel $W_2$, it is apparent that 
\begin{eqnarray*}
R_{HB}(D,d_{\max}|W_*,W_2) \ge R(D)
\end{eqnarray*}
where $R(D)$ is the rate-distortion function for one-to-one lossy coding without side information.
Hence, if $E \ge E_*$, we have
\begin{eqnarray*}
R_a(D|E) \ge R(D).
\end{eqnarray*}
Since $R_a(D|E) \le R(D)$ always holds, we have the following corollary.
\begin{corollary}
If $E \ge E_*$, then we have
\begin{eqnarray*}
R_a(D|E) = R(D).
\end{eqnarray*}
\end{corollary}

%%%% Binary Hamming %%%%%
\subsection{Binary Hamming Example}
\label{section:binary-hamming}

To provide some insight on our results, we consider the
binary Hamming example, i.e., we assume that
${\cal X} = {\cal Y} = \hat{{\cal X}} = \{0,1\}$, $P_X(0) = P_X(1) = \frac{1}{2}$, and
\begin{eqnarray*}
e(x,y) &=& \left\{ 
\begin{array}{ll}
0 & \mbox{if } x = y \\
1 & \mbox{else}
\end{array}
\right., \\
d(x,\hat{x}) &=& \left\{
\begin{array}{ll}
0 & \mbox{if } x = \hat{x} \\
1 & \mbox{else}
\end{array}
\right..
\end{eqnarray*}
In this section, we assume that $E \le \frac{1}{2}$.

We first consider the maximum distortion class.
In this case, the set ${\cal W}_1(E)$ can be parametrized by
two parameters $(\alpha, \beta)$ satisfying
\begin{eqnarray*}
\frac{\alpha + \beta}{2} \le E
\end{eqnarray*}
(see Fig.~\ref{Fig:binary-channel} and Fig.~\ref{Fig:admissible-channel}).

By the concavity of $\tilde{R}_{WZ}(D|W,E)$ with respect to $W$ (Lemma \ref{lemma:cocavity-pseudo})
and by the symmetry with respect to $\alpha$ and $\beta$, we have
\begin{eqnarray*}
\argmax_{W \in {\cal W}_1(E)} \tilde{R}_{WZ}(D|W,E) = \mbox{BSC}(E).
\end{eqnarray*}

Let $0,1 \in {\cal U}$ be constant functions that output $0$ or $1$ 
irrespective of $y$ and let $y$ be the function that output $y$ itself.
Similarly, let $\bar{y}$ be the function that outputs $y \oplus 1$.
In the binary Hamming case, ${\cal U} = \{0,1,y,\bar{y} \}$.
For $W^* = \mbox{BSC}(E)$, it is known that
\begin{eqnarray*}
R_{WZ}(D|W^*) = \min_{V \in {\cal V}(W^*,D)} \phi(V,W^*)
\end{eqnarray*}
is achieved by the test channel of the form
\begin{eqnarray*}
\hat{V}(u|x) = \left\{
\begin{array}{ll}
\lambda (1-q) & \mbox{if } u = x \\
\lambda q & \mbox{if } u = x \oplus 1 \\
(1 - \lambda) & \mbox{if } u = y
\end{array}
\right.
\end{eqnarray*}
for some $0 \le \lambda \le 1$ and $0 \le q \le \frac{1}{2}$
($\lambda$ represents the time sharing).
In this case, the distortion is given by
\begin{eqnarray*}
\lefteqn{
\lambda \sum_{\hat{x},x} P_X(x) V_q(\hat{x}|x) d(x,\hat{x}) }  \\
&& + (1-\lambda) \sum_{x,y} P_X(x) W^*(y|x) d(x,y)  \\
&=& \lambda \sum_{\hat{x},x} P_X(x) V_q(\hat{x}|x) d(x,\hat{x}) + (1 - \lambda) E \\
&\le& D,
\end{eqnarray*}
where $V_q = \rom{BSC}(q)$.
Since every channel $W \in {\cal W}_1(E)$ satisfies
\begin{eqnarray*}
\lefteqn{ \sum_{x,y} P_X(x) W(y|x) d(x,y) } \\
&=& \sum_{x,y} P_X(x) W(y|x) e(x,y) \\
&\le& E,
\end{eqnarray*}
we find that $\hat{V} \in {\cal V}(E,D)$. Thus, the matching condition
of Corollary \ref{corollary:matching-1} is satisfied for this binary Hamming example.

%%%%%%%% Fig %%%%%%%%%%%%%%%
\begin{figure}[t]
\centering
\includegraphics[width=0.5\linewidth]{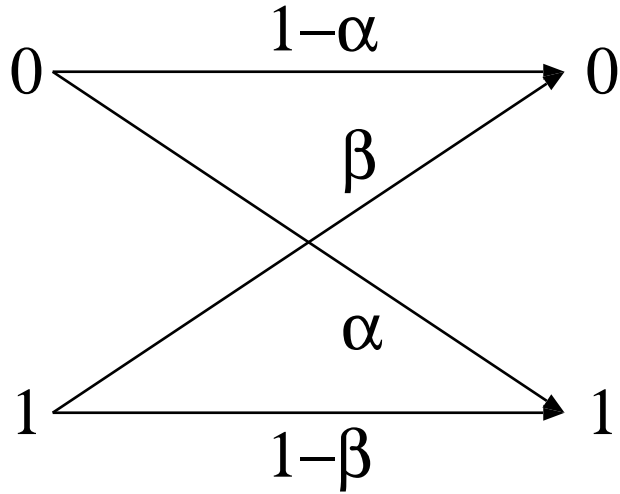}
\caption{Parametrization of binary channels.}
\label{Fig:binary-channel}
\end{figure}
%%%%%%%% Fig %%%%%%%%%%%%%%%
\begin{figure}[t]
\centering
\includegraphics[width=0.5\linewidth]{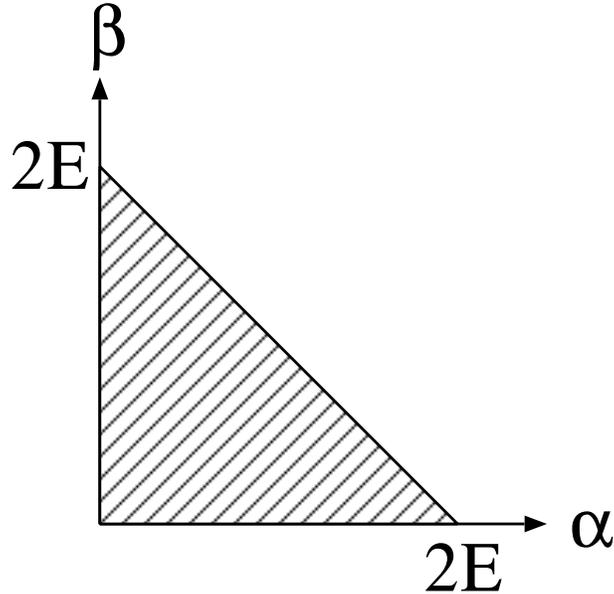}
\caption{The set of all channels in ${\cal W}_1(E)$.}
\label{Fig:admissible-channel}
\end{figure}

Next, we consider the average distortion class.
We evaluate the upper bound (\ref{eq:rd-average-class}).
We first fix $W^*$ to be $\mbox{BSC}(E)$. Note that 
\begin{eqnarray}
\min_{V \in {\cal V}(E,D)} I(P_X,V) \ge \min_{V \in {\cal V}(W^*,D)} I(P_X,V).
\label{eq:objective-of-optimization}
\end{eqnarray} 
For a test channel $V \in {\cal V}(W^*,D)$, let $\bar{V}$ be a test channel 
such that $\bar{V}(u | x) = V( u \oplus 1 | x \oplus 1)$ for $u \in \bar{{\cal U}}$
and $\bar{V}(u|x) = V(u|x \oplus 1)$ for $u \in \{ y, \bar{y} \}$.
Then, by the symmetry of the BSC and the source $P_X$, we have $\bar{V} \in {\cal V}(W^*,D)$ and
$I(P_X, V) = I(P_X,\bar{V})$.
 By the convexity of the mutual information for channel, we have
\begin{eqnarray*}
I(P_X,\tilde{V}) \le \frac{1}{2} I(P_X,V) + \frac{1}{2} I(P_X,\bar{V}),
\end{eqnarray*}
where $\tilde{V} = \frac{1}{2} V + \frac{1}{2} \bar{V}$.
This means that the minimum in the right hand side of (\ref{eq:objective-of-optimization})
is achieved by a symmetric test channel, i.e., 
$V(u|x) = V(u \oplus 1 | x \oplus 1)$ for $u \in \bar{{\cal U}}$
 and $V(u | x) = V(u | x \oplus 1)$ for $u \in \{y, \bar{y} \}$. 
 Furthermore, for $E \le \frac{1}{2}$, we can assume that 
 $V(\bar{y} |x) = 0$ because using $\bar{y}$ only makes the distortion larger.
 We also note that such a symmetric test channel satisfies $V \in {\cal V}(E,D)$.
 Thus, the equality in (\ref{eq:objective-of-optimization}) actually holds.
Consequently, the upper bound on $R_a(D|E)$ in this example is the time sharing between
 the ordinary rate-distortion function and the distortion that can be achieved only
 by the estimation, i.e., the point $(E,0)$.

%%%%
\subsection{Discussion on Universality}

In this section, we discuss on the definitions of the
universal Wyner-Ziv coding. We also discuss the relation
between the universal Wyner-Ziv coding and the Heegard-Berger
problem. 

Let us consider the binary Hamming case as in the previous section.
Let $X$ be the uniform random variable on $\{0,1\}$. Let $W_1$ be 
the binary channel in Fig.~\ref{Fig:binary-channel} with
$\alpha = 2 E$ and $\beta = 0$, and let $W_2$ be the 
binary channel in Fig.~\ref{Fig:binary-channel} with $\alpha = 0$ and $\beta = 2 E$.
Obviously, the Wyner-Ziv rate-distortion functions for $W_1$ and $W_2$ coincide, i.e.,
\begin{eqnarray*}
R_{WZ}(D|W_1) = R_{WZ}(D|W_2). 
\end{eqnarray*}
It should be also noted that ${\cal W}_1(E)$ is the convex hull of
the set $\{ W_1, W_2 \}$.

As we have mentioned in Section \ref{section:problem},
in the ordinary definition of the universality,
we require that there exists a universal code that works well for 
every ${\cal W}_{WZ}(R,D)$ instead of ${\cal W}_1(E)$.
If we set $R = R_{WZ}(D|W_1) = R_{WZ}(D|W_2)$, then we have
\begin{eqnarray*}
W_1, W_2 \in {\cal W}_{WZ}(R,D).
\end{eqnarray*}
Thus, at least, we have to construct a code that is universal for both
$W_1$ and $W_2$, which can be regarded as a special case of
the Heegard-Berger problem \cite{heegard:85}.
The rate-distortion function $R_{HB}(D,D|W_1,W_2)$
is not known, but we have a trivial lower bound
\begin{eqnarray}
\lefteqn{ R_{HB}(D,D|W_1,W_2) } \\
&\ge& R_{WZ}(D|W_1) = R_{WZ}(D|W_2). 
\label{eq:trivial-lower-bound}
\end{eqnarray}
The equality in (\ref{eq:trivial-lower-bound}) is a required condition
such that the universal coding in the sense of ${\cal W}_{WZ}(R,D)$ to be possible.
In other word, if the strict inequality holds in (\ref{eq:trivial-lower-bound}),
this means that the universal coding in the sense of ${\cal W}_{WZ}(R,D)$ is 
impossible. Showing whether the equality holds or not is an important open problem.

A straightforward upper bound on 
$R_{HB}(D,D|W_1,W_2)$ can be derived as follows.
Let $V_s \in {\cal P}({\cal U}|{\cal X})$ be a symmetric test channel such that
\begin{eqnarray*}
V_s(0|0) &=& V_s(1|1), \\
V_s(y|0) &=& V_s(y|1), \\
V_s(\bar{y}|0) &=& V_s(\bar{y}|1) = 0.
\end{eqnarray*} 
Then, by taking $V_s(0|0)$ appropriately, we have
\begin{eqnarray*}
V_s \in {\cal V}(W_1,D) \cap {\cal V}(W_2,D).
\end{eqnarray*}
The achievability of 
\begin{eqnarray*}
\phi(V_s, W_1) = \phi(V_s, W_2)
\end{eqnarray*}
can be also derived from the known upper bound in \cite{heegard:85}.
Thus, we have
\begin{eqnarray*}
\lefteqn{ R_{HB}(D,D|W_1,W_2) } \\
&\le& \tilde{R}_{HB}(D,D|W_1,W_2) \\
&:=& \min_{V_s \in {\cal V}(W_1,D) \cap {\cal V}(W_2,D)} \phi(V_s, W_1) 
\end{eqnarray*}
Numerical calculations of $\tilde{R}_{HB}(D,D|W_1,W_2)$
and $R_{WZ}(D|W_1)$ are compared in Fig.~\ref{Fig:comparison-of-rd}.
For comparison, we also plotted $R_m(D|E)$ in the figure.
As we can find from the figure, $R_m(D|E)$ is much larger than
$\tilde{R}_{HB}(D,D|W_1,W_2)$. This is because 
${\cal W}_1(E)$ involves $\mbox{BSC}(E)$.

%%%%%%%% Fig %%%%%%%%%%%%%%%
\begin{figure}[t]
\centering
\includegraphics[width=\linewidth]{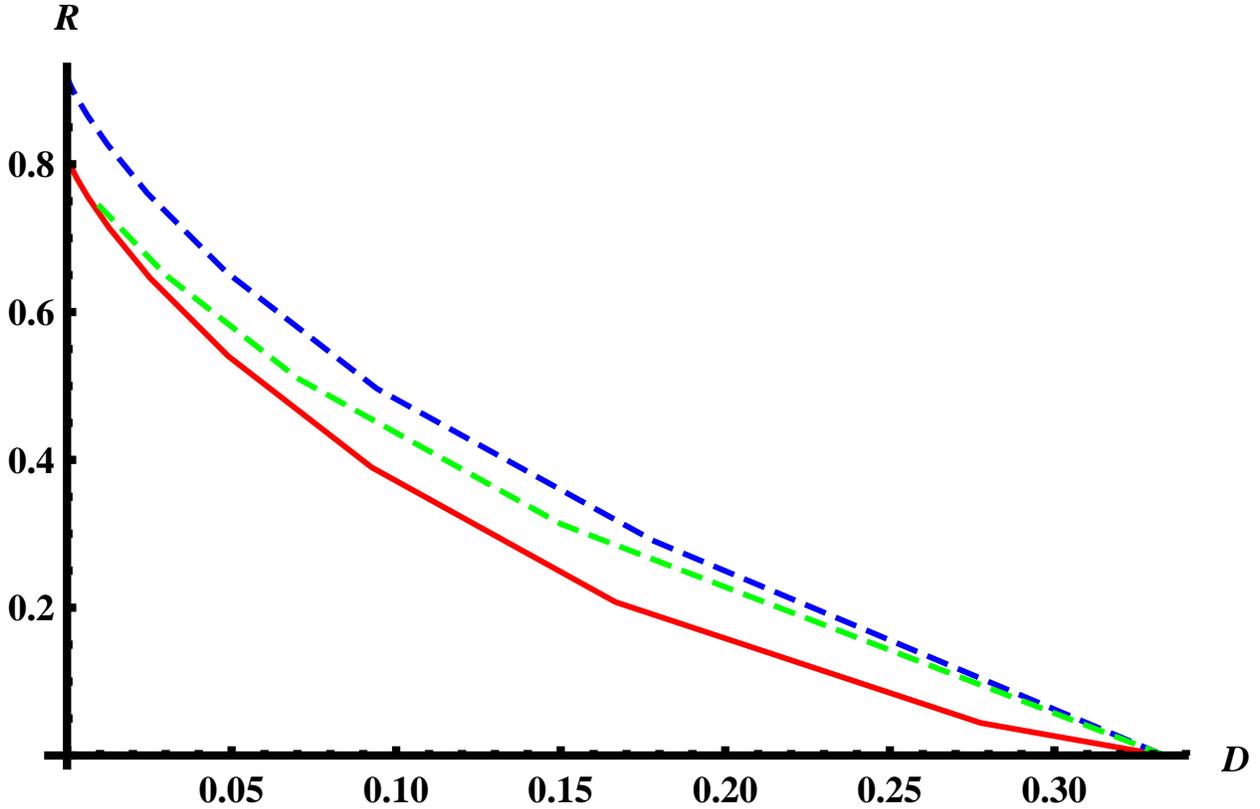}
\caption{Comparison among $\tilde{R}_{HB}(D,D|W_1,W_2)$,
$R_{WZ}(D|W_1)$, and $R_m(D|E)$. The red solid line is $R_{WZ}(D|W_1)$.
The green dashed line is $\tilde{R}_{HB}(D,D|W_1,W_2)$.
The blue dashed line is $R_m(D|E)$.}
\label{Fig:comparison-of-rd}
\end{figure}

%%%% Proof of Main Theorem %%%%
\section{Proof of Theorem \ref{theorem:main}}
\label{section:proof-of-theorem:main}

%%% Converse %%%%
\subsection{Proof of Converse Part}
\label{subset:converse-proof-maximum}

First we consider the case with $E=0$. Let $W \in {\cal W}_1(0)$.
Then, from the definition of ${\cal W}_m(E)$, we have $\{ W^{\times n} \}_{n=1}^{\infty} \in {\cal W}_m(0)$,
which implies 
\begin{eqnarray*}
R_m(D|0) \ge \max_{W \in {\cal W}_1(0)} R_{WZ}(D|W).
\end{eqnarray*}

Next, we consider the case with $E>0$ and $D>0$. For any $0 < \delta < E$, let $W \in {\cal W}_1(E-\delta)$. Then,
from the definition of ${\cal W}_m(E)$, we have $\{ W^{\times n} \}_{n=1}^{\infty} \in {\cal W}_m(E)$,
which implies 
\begin{eqnarray*}
R_m(D|E) \ge \max_{W \in {\cal W}_1(E-\delta)} R_{WZ}(D|W).
\end{eqnarray*}
Since this inequality holds for arbitrary $0 < \delta < E$
and $R_{WZ}(D|W)$ is continuous with respect to $W$ for $D>0$,
which will be proved in Appendix \ref{appendix:continuity-of-RWZ},
we have (\ref{eq:converse-1}).

When $E > 0$ and $D=0$, for $\epsilon >0$, we first prove
\begin{eqnarray*}
R_m(0|E) &\ge& \max_{W \in {\cal W}_1(E-\delta)} R_{WZ}(0|W) \\
&\ge& \max_{W \in {\cal W}_1(E-\delta)} R_{WZ}(\epsilon|W).
\end{eqnarray*}
Then, by the continuity argument, we have
\begin{eqnarray*}
R_m(0|E) \ge  \lim_{\epsilon \downarrow 0} \max_{W \in {\cal W}_1(E)} R_{WZ}(\epsilon|W).
\end{eqnarray*}
\qed

\subsection{Proof of Direct Part}
\label{subsection:direct-part}

Note that the function $\phi(\cdot, W)$ is a convex function for fixed $W$,
$\phi(V,\cdot)$ is a concave function for fixed $V$, and 
${\cal W}_1(E)$ and ${\cal V}(E,D)$ are convex sets.
Thus, (\ref{eq:direct-2}) is derived from (\ref{eq:direct-1})
by applying the saddle point theorem \cite{bertsekas:03}.
We prove (\ref{eq:direct-1}) by three steps.
First, we prove that there exists a universal code for i.i.d. channels.
Then, we show that there exists a randomized universal code for permutation invariant channels.
Finally, we de-randomize the randomized universal code by using the technique of \cite{ahlswede:80,ziv:84}.

%%% iid %%%% 
\subsubsection{Code for i.i.d. Channel}
\label{subsection-iid}

In this section, we construct a universal Wyner-Ziv code 
for a fixed test channel such that it works well for
every $W \in {\cal W}_1(E) \cap {\cal P}_n({\cal Y}|{\cal X})$.
We construct a universal Wyner-Ziv code by using the 
output statistics of random binning argument 
recently introduced by \cite{yassaee:12}.
We note that a universal Wyner-Ziv code can be also
constructed from the coding method in \cite{kelly:12}.

Let us fix $V \in {\cal V}(D,E)$.
We use two kinds of bin codings $f_n : {\cal U}^n \to {\cal S}_n$
and $g_n: {\cal U}^n \to {\cal L}_n$. Let $F_n$ and $G_n$ be
random bin codings. 
For arbitrary small $\delta > 0$, 
let $R_f, R_g \ge 0$ be the real numbers such that
\begin{eqnarray}
R_f &=& H(U|X) - \delta, \label{eq:pa-rate} \\
R_g &=& \max_{W \in {\cal W}_1(E)} \phi(V,W) + 2 \delta \\
&=& \max_{W \in {\cal W}_1(E)} I(U;X|Y) + 2 \delta. 
\end{eqnarray}
Since
\begin{eqnarray*}
I(U;X|Y) &=& H(U|Y) - H(U|X,Y) \\
&=& H(U|Y) - H(U|X), 
\end{eqnarray*}
we have
\begin{eqnarray}
R_f + R_g = \max_{W \in {\cal W}_1(E)} H(U|Y) + \delta.
\label{eq:rate-sum}
\end{eqnarray}
Let $|{\cal S}_n| = \lfloor 2^{n R_f} \rfloor$ and 
$|{\cal L}_n| = \lceil 2^{n R_g} \rceil$.

From (\ref{eq:rate-sum}), we find that the sum rate $R_f + R_g$ is sufficiently
large for the Slepian-Wolf coding.
We use the following lemma on universal Slepian-Wolf coding.
\begin{lemma}
\label{lemma:universal-sw}
For sufficiently large $n$, there exists $\mu_1 > 0$ and a universal decoder
$\kappa_n:{\cal Y}^n \times {\cal S}_n \times {\cal L}_n \to {\cal U}^n$ such that
\begin{eqnarray*}
\mathbb{E}_{F_n G_n}\left[ P_{err}(F_n,G_n,W) \right] \le 2^{- \mu_1 n}
\end{eqnarray*}
for every $W \in {\cal W}_1(E) \cap {\cal P}_n({\cal Y}|{\cal X})$, where 
$P_{err}(F_n,G_n,W)$ is the error probability of the Slepian-Wolf coding for
channel $W$ when the bin codings $(F_n,G_n)$ are used.
\end{lemma}
\begin{proof}
The lemma is proved exactly in the same manner as \cite{csiszar:82}.
A few modifications are that we use the random bin coding instead of
the random linear coding\footnote{we can also use the random linear coding
instead of the random bin coding because Lemma \ref{lemma:privacy-amplification} 
holds under the condition that $f_n$ is chosen from a universal hash family and
the random linear coding ensemble is a universal hash family.}, and that we evaluate
the ensemble average of the error probability.
\end{proof}

%%%
From (\ref{eq:pa-rate}), we find that the rate $R_f$ is sufficiently small
to generate the uniform random variable that is independent of $X^n$.
We use the privacy amplification lemma (Lemma \ref{lemma:privacy-amplification})
described in Appendix \ref{proof-of-lemma:privacy-amplification}.

We construct a code as follows.
Let 
\begin{eqnarray*}
P_{\hat{U}^n|Y^n S_n L_n}(u^n | y^n,s_n,\ell_n) = \bol{1}[u^n = \kappa_n(y^n,s_n,\ell_n)]
\end{eqnarray*}
be the distribution describing the Slepian-Wolf decoder.
Let
\begin{eqnarray*}
\lefteqn{ P_{S_n L_n U^n X^n Y^n \hat{U}^n}(s_n,\ell_n,u^n,x^n,y^n,\hat{u}^n) } \\
&=& P_{S_n X^n}(s_n,x^n) P_{U^n|S_n X^n}(u^n|s_n,x^n) P_{L_n|U^n}(\ell_n|u^n) \\
&& P_{Y^n|X^n}(y^n|x^n) P_{\hat{U}^n|Y^n S_n L_n}(\hat{u}^n|y^n,s_n,\ell_n)
\end{eqnarray*}
and
\begin{eqnarray*}
\lefteqn{ \hat{P}_{\bar{S}_n L_n U^n X^n Y^n \hat{U}^n}(s_n,\ell_n,u^n,x^n,y^n,\hat{u}^n) } \\
&=& P_{\bar{S}_n}(s_n) P_{X^n}(x^n) P_{U^n|S_n X^n}(u^n|s_n,x^n) P_{L_n|U^n}(\ell_n|u^n) \\
&& P_{Y^n|X^n}(y^n|x^n) P_{\hat{U}^n|Y^n S_n L_n}(\hat{u}^n|y^n,s_n,\ell_n).
\end{eqnarray*}
The distribution $P_{S_n L_n U^n X^n Y^n \hat{U}^n}$ describes a virtual
coding scheme in which the encoder sends both $F_n(U^n)$ and $G_n(U^n)$.
The distribution $\hat{P}_{\bar{S}_n L_n U^n X^n Y^n \hat{U}^n}$ describes a real
coding scheme in which the encoder sends only $G_n(U^n)$ and uses the common randomness
$\bar{S}_n$ that is shared with the decoder.
Note that $P_{U^n|S_n X^n}(u^n|s_n,x^n)$ is a randomized quantizer,
which is derived from the bin coding $f_n$ and the test channel $P_{U|X}$ via
$P_{U^n X^n S_n}$.
From Lemma \ref{lemma:privacy-amplification} and the fact that the variational
distance does not increase by data processing or marginalization, we have
\begin{eqnarray*}
\mathbb{E}_{F_nG_n}\left[ \| \hat{P}_{\bar{S}_n U^n X^n Y^n \hat{U}^n} - P_{S_n U^n X^n Y^n \hat{U}^n} \| \right] 
\le 2^{- \mu_2 n}
\end{eqnarray*}
for some $\mu_2 > 0$.
By the large deviation bound such as the Bernstein inequality, there exists $\mu_3 > 0$ such that
\begin{eqnarray}
P_{U^n X^nY^n}(\{ d_n(x^n,u^n(y^n)) > D + \delta\}) \le 2^{- \mu_3 n}.
\label{eq:distortion-large-deviation}
\end{eqnarray}
It should be noted that the bound (\ref{eq:distortion-large-deviation}) is uniform
with respect to the channel $W$.
By Lemma \ref{lemma:universal-sw}, we have
\begin{eqnarray*}
\lefteqn{ \mathbb{E}_{F_n G_n}[ P_{S_n U^n X^n Y^n \hat{U}^n}(\{ d_n(x^n, \hat{u}^n(y^n)) > D + \delta \}) ] } \\
&\le& \mathbb{E}_{F_n G_n}[ P_{S_n U^n X^n Y^n \hat{U}^n}(\{ d_n(x^n, u^n(y^n)) > D + \delta  \\
&& ~~~~   \mbox{ or } u^n \neq \hat{u}^n \}) ] \\
&\le& 2^{- \mu_1 n} + 2^{-\mu_3 n}.
\end{eqnarray*}
Since 
\begin{eqnarray*}
\lefteqn{ \hat{P}_{\bar{S}_n U^n X^n Y^n \hat{U}^n}({\cal A}) - P_{S_n U^n X^n Y^n \hat{U}^n}({\cal A}) } \\
&\le& \| \hat{P}_{\bar{S}_n U^n X^n Y^n \hat{U}^n} - P_{S_n U^n X^n Y^n \hat{U}^n} \| 
\end{eqnarray*}
for any set ${\cal A}$, we have
\begin{eqnarray*}
&& \hspace{-10mm} \mathbb{E}_{F_n G_n}[ \hat{P}_{\bar{S}_n U^n X^n Y^n \hat{U}^n}(\{ d_n(x^n, \hat{u}^n(y^n)) > D + \delta \})] \\
&\le& 3 \cdot 2^{-n \min \mu_i}.
\end{eqnarray*}
Since 
\begin{eqnarray*}
|{\cal W}_1(E) \cap {\cal P}_n({\cal Y}|{\cal X}) | \le (n+1)^{|{\cal X}||{\cal Y}|},
\end{eqnarray*}
there exists at least one realization $(f_n,g_n,s_n)$
of $(F_n,G_n,S_n)$ such that
\begin{eqnarray*}
&& \hspace{-10mm} \hat{P}_{U^n X^n Y^n \hat{U}^n | \bar{S}_n}(\{ d_n(x^n, \hat{u}^n(y^n)) > D + \delta \} | s_n)  \\
&\le& 3 (n+1)^{|{\cal X}||{\cal Y}|}2^{-n \min \mu_i}.
\end{eqnarray*}
for every $W \in {\cal W}_1(E) \cap {\cal P}_n({\cal Y}|{\cal X})$. 
Furthermore, let $K_n$ be a random variable that simulate the randomized
quantizer $P_{U^n|S_n X^n}$. Then, we can also eliminate this 
randomness in a similar manner as above.

In summary, we have shown the following.
\begin{lemma}
\label{lemma:universal-for-iid}
For any $V \in {\cal V}(E,D)$ and any $\delta > 0$, there exists a universal code
$(\varphi_n,\psi_n)$ and a constant $\mu > 0$ such that
\begin{eqnarray*}
\frac{1}{n} \log |{\cal M}_n| \le \max_{W \in {\cal W}_1(E)} \phi(V,W) +  2 \delta
\end{eqnarray*}
and
\begin{eqnarray*}
\Pr\{ d_n(X^n, \psi_n(\varphi_n(X^n), Y^n)) > D + \delta \} \le 2^{- \mu n}
\end{eqnarray*}
for every $W \in {\cal W}_1(E) \cap {\cal P}_n({\cal Y}|{\cal X})$ provided that $n$ is sufficiently large.
\end{lemma}

%%% permutation invariant %%%
\subsubsection{Code for Permutation Invariant Channel}

In Section \ref{subsection-iid}, we constructed a universal Wyner-Ziv code 
$(\varphi_n,\psi_n)$ for a fixed test channel such that it works well for
every $W \in {\cal W}_1(E) \cap {\cal P}_n({\cal Y}|{\cal X})$.
In this section, we use this code to the channel in ${\cal W}_m(E)$.
Let $\pi_n$ be random permutation on $\{1,\ldots,n \}$.
We first apply the random permutation to the sequence $(X^n,Y^n)$
and then use $(\varphi_n,\psi_n)$. It should be noted that the encoder and
the decoder agree with a realization of the random permutation in this section.
We denote 
\begin{eqnarray*}
P_X^n \cdot W^n(x^n,y^n) = P_X^n(x^n) W^n(y^n|x^n).
\end{eqnarray*}
Note that
\begin{eqnarray*}
\lefteqn{ \mathbb{E}_{\pi_n}\left[ P_X^n \cdot W^n(\pi_n(x^n), \pi_n(y^n)) \right] } \\
&=& \mathbb{E}_{\pi_n}\left[ P_X^n(\pi_n(x^n) W^n(\pi_n(y^n) | \pi_n(x^n)) \right] \\
&=& P_X^n(x^n) \mathbb{E}_{\pi_n} \left[ W^n(\pi_n(y^n) | \pi_n(x^n)) \right],
\end{eqnarray*}
and we consider the average performance with respect to the permutation.
Thus, without loss of generality, we can assume that $W^n$ is permutation
invariant, i.e., $W^n(y^n|x^n) = W^n(\tilde{y}^n|\tilde{x}^n)$ if
$P_{x^n y^n} = P_{\tilde{x}^n \tilde{y}^n}$.

\begin{lemma}
\label{lemma:permutation-invariant-probability}
Let ${\cal A}_n \subset {\cal X}^n \times {\cal Y}^n$. Suppose that 
\begin{eqnarray*}
P_X^n \cdot \bar{W}^{\times n}({\cal A}_n^c) \le \bar{\varepsilon} 
\end{eqnarray*}
for every $\bar{W} \in {\cal W}_1(E + \delta e_{\max}) \cap {\cal P}_n({\cal Y}|{\cal X})$.
Then, for any conditional type $\bar{W} \in \bar{{\cal W}}_n(T_{X,\delta},E)$, we have
\begin{eqnarray*}
\lefteqn{ \sum_{x^n \in T_{X,\delta}} P_X^n(x^n) \bol{1}[\bar{W} \in {\cal P}_n({\cal Y}|P_{x^n})] } \\ 
&&	\times \sum_{y^n \in T_{\bar{W}}(x^n)} \frac{1}{|T_{\bar{W}}(x^n)|} \bol{1}[(x^n,y^n) \in {\cal A}_n^c] \\
	&\le& (n+1)^{|{\cal X}| |{\cal Y}|} \bar{\varepsilon}. 
\end{eqnarray*}
\end{lemma}
\begin{proof}
For any $\bar{W} \in {\cal P}_n({\cal Y}|P_{x^n})$, note that
\begin{eqnarray}
\bar{W}^{\times n}(T_{\bar{W}}(x^n) | x^n) &\ge& \frac{1}{(n+1)^{|{\cal X}||{\cal Y}|}} 2^{-n D(\bar{W} \| \bar{W} |P_{x^n})} \\
	&=& \frac{1}{(n+1)^{|{\cal X}||{\cal Y}|}}.
	\label{eq:type-channel-probability}
\end{eqnarray}
From (\ref{eq:typicality-implication}), we have
\begin{eqnarray*}
\bar{W} \in \bar{{\cal W}}_n(T_{X,\delta},E) \Longrightarrow \bar{W} \in {\cal W}_1(E + \delta e_{\max}).
\end{eqnarray*}
Thus, for every  
$\bar{W} \in \bar{{\cal W}}_n(T_{X,\delta},E)$ we have
\begin{eqnarray*}
\bar{\varepsilon} &\ge& P_X^n \cdot \bar{W}^{\times n}({\cal A}_n^c) \\
	&\ge& \sum_{x^n \in T_{X,\delta}} P_X^n(x^n) \bol{1}[\bar{W} \in {\cal P}_n({\cal Y}|P_{x^n})] 
	\sum_{y^n \in T_{\bar{W}}(x^n)} \\
	&& \times \bar{W}^{\times n}(T_{\bar{W}}(x^n) |x^n) \frac{1}{|T_{\bar{W}}(x^n)|} \bol{1}[(x^n,y^n) \in {\cal A}_n^c] \\
	&\ge& \frac{1}{(n+1)^{|{\cal X}||{\cal Y}|}} \sum_{x^n \in T_{X,\delta}} P_X^n(x^n) \bol{1}[\bar{W} \in {\cal P}_n({\cal Y}|P_{x^n})] \\
	&&\times 	\sum_{y^n \in T_{\bar{W}}(x^n)} \frac{1}{|T_{\bar{W}}(x^n)|} \bol{1}[(x^n,y^n) \in {\cal A}_n^c], 
\end{eqnarray*}
which implies the statement of the lemma.
\end{proof}

\begin{lemma}
\label{lemma:code-for-permutation-invariant}
Suppose that the code $(\varphi_n, \psi_n)$ satisfies 
\begin{eqnarray*}
\Pr\left\{ d_n(X^n, \psi_n(\varphi_n(X^n),Y^n)) > D + \delta \right\} \le \bar{\varepsilon}
\end{eqnarray*}
for every i.i.d. channel $\bar{W}^{\times n}$ such that
$\bar{W} \in {\cal W}_1(E + \delta e_{\max}) \cap {\cal P}_n({\cal Y}|{\cal X})$, where 
$(X^n,Y^n) \sim P_X^n \cdot \bar{W}^{\times n}$.
Then, we have 
\begin{eqnarray*}
\lefteqn{ \Pr\left\{ d_n(X^n, \psi_n(\varphi_n(X^n),Y^n)) > D + \delta \right\} } \\
	&\le& (n+1)^{2 |{\cal X}| |{\cal Y}|} \bar{\varepsilon}
	+ P_X^n((T_{X,\delta}^n)^c) + \delta_1 
\end{eqnarray*}
for every permutation invariant (not necessarily i.i.d.) $W^n$ satisfying
\begin{eqnarray*}
\Pr\left\{ e_n(X^n,Y^n) > E \right\} \le \delta_1,
\end{eqnarray*}
where $(X^n,Y^n) \sim P_X^n \cdot W^n$.
\end{lemma}
\begin{proof}
Suppose that $(X^n,Y^n) \sim P_X^n \cdot W^n$.
By using Lemma \ref{lemma:permutation-invariant-probability} for 
\begin{eqnarray*}
{\cal A}_n := \left\{ (x^n,y^n): d_n(x^n, \psi_n(\varphi_n(x^n),y^n)) \le D + \delta \right\},
\end{eqnarray*}
we have
\begin{eqnarray*}
\lefteqn{ \Pr\left\{ d_n(X^n, \psi_n(\varphi_n(X^n),Y^n)) > D + \delta \right\} }  \\
&\le& \sum_{x^n \notin T_{X,\delta}} P_X^n(x^n) \\
&& + \sum_{\bar{W} \in \bar{{\cal W}}_n(T_{X,\delta},E)} \sum_{x^n \in T_{X,\delta}} P_X^n(x^n) \bol{1}[\bar{W} \in {\cal P}_n({\cal Y}|P_{x^n})] \\
&& \times 	\sum_{y^n \in T_{\bar{W}}(x^n)} W^n(T_{\bar{W}}(x^n) |x^n) \frac{1}{|T_{\bar{W}}(x^n)|} \bol{1}[(x^n,y^n) \in {\cal A}_n^c] \\
&& + \sum_{\bar{W} \notin \bar{{\cal W}}_n(T_{X,\delta},E)} \sum_{x^n \in T_{X,\delta}} P_X^n(x^n) \bol{1}[\bar{W} \in {\cal P}_n({\cal Y}|P_{x^n})] 
 W^n(T_{\bar{W}}(x^n) |x^n) \\
&\le& P_X^n(T_{X,\delta}^c) \\
&& + \sum_{\bar{W} \in \bar{{\cal W}}_n(T_{X,\delta},E)} \sum_{x^n \in T_{X,\delta}} P_X^n(x^n) \bol{1}[\bar{W} \in {\cal P}_n({\cal Y}|P_{x^n})] \\
&& \times	\sum_{y^n \in T_{\bar{W}}(x^n)}  \frac{1}{|T_{\bar{W}}(x^n)|} \bol{1}[(x^n,y^n) \in {\cal A}_n^c] \\
&&+ \Pr\left\{ e_n(X^n,Y^n) > E \right\} \\
&\le& (n+1)^{2 |{\cal X}| |{\cal Y}|} \bar{\varepsilon}
	+ P_X^n(T_{X,\delta}^c) + \delta_1,
\end{eqnarray*}
where we used $W^n(T_{\bar{W}}(x^n) | x^n) \le 1$ to bound the second term, we used 
the fact 
\begin{eqnarray*}
\bar{W} \notin \bar{{\cal W}}_n(P_{x^n},E) \Longleftrightarrow e(P_{x^n}, \bar{W}) > E
\end{eqnarray*}
and \eqref{eq:distortion-empirical-dependency}
to bound the third term in the second inequality, and we used Lemma \ref{lemma:permutation-invariant-probability} to bound the second term 
in the third inequality.
\end{proof}

By combining Lemma \ref{lemma:universal-for-iid} and 
Lemma \ref{lemma:code-for-permutation-invariant} and by noting
the definition of ${\cal W}_m(E)$, we have the following.
\begin{lemma}
\label{lemma:code-with-random-permutation}
For any $V \in {\cal V}(E + \delta e_{\max},D)$, any $\delta > 0$,
and any $\varepsilon > 0$, there exists a universal code
$(\varphi_n,\psi_n)$ such that
\begin{eqnarray*}
\frac{1}{n} \log |{\cal M}_n| \le \max_{W \in {\cal W}_1(E + \delta e_{\max})} \phi(V,W) +  2 \delta
\end{eqnarray*}
and
\begin{eqnarray}
&&\hspace{-10mm} \mathbb{E}_{\pi_n}\left[ \Pr\{ d_n(\pi_n(X^n), \psi_n(\varphi_n(\pi_n(X^n)), \pi_n(Y^n))) > D + \delta \} \right]  \nonumber \\
&\le& \varepsilon
\label{eq:permutation-averaged}
\end{eqnarray}
for every $\bol{W} \in {\cal W}_m(E)$ provided that $n$ is sufficiently large.
\end{lemma}

%%% Derandomization %%%
\subsubsection{De-Randomization}

Now we reduce the size of random permutation by using the de-randomization
technique.
\begin{lemma}
\label{lemma:after-de-randomize}
Suppose that $(\varphi_n,\psi_n)$ satisfies (\ref{eq:permutation-averaged}).
Then, for arbitrary $\delta_2, \gamma > 0$, there exists
$m_n = 2^{\delta_2 n}$ permutations $\{\pi_n^{(1)},\ldots,\pi_n^{(m_n)} \}$ such that
\begin{eqnarray*}
&&\hspace{-10mm}  \frac{1}{m_n} \sum_{i=1}^{m_n}  \Pr\{ d_n(\pi_n^{(i)}(X^n), \psi_n(\varphi_n(\pi_n^{(i)}(X^n)), \pi_n^{(i)}(Y^n)))  \\
&& ~~~~~~~> D + \delta \} \le \varepsilon + \gamma 
\end{eqnarray*}
provided that $n$ is sufficiently large.
\end{lemma}
\begin{proof}
For a permutation $\pi_n$ and $(x^n,y^n) \in {\cal X}^n \times {\cal Y}^n$, we denote
\begin{eqnarray*}
\lefteqn{ I(\pi_n,x^n,y^n) }\\
&=& \bol{1}[ d_n(\pi_n(x^n), \psi_n(\varphi_n(\pi_n(x^n)), \pi_n(y^n))) > D + \delta ].
\end{eqnarray*}
Let $\pi_n^{(1)}, \ldots,\pi_n^{(m_n)}$ be randomly generated permutations, and
let $\bar{I}(x^n,y^n) = \mathbb{E}_{\pi_n}[I(\pi_n,x^n,y^n)]$.
Then, by using Lemma \ref{lemma:bernstein-trick} for $A_i = I(\pi_n^{(i)},x^n,y^n)$, 
$b=1$, and $\alpha = \frac{\gamma}{2}$, we have
\begin{eqnarray*}
&& \hspace{-10mm} \Pr\left\{ \frac{1}{m_n}\sum_{i=1}^{m_n} I(\pi_n^{(i)},x^n,y^n) \ge \bar{I}(x^n,y^n) + \gamma \right\} \\
&\le& \exp\{-(\gamma^2/4) m_n\}. 
\end{eqnarray*}
Furthermore, by using the union bound, we have
\begin{eqnarray}
&& \hspace{-10mm} \Pr\left\{ \exists (x^n,y^n)~\frac{1}{m_n}\sum_{i=1}^{m_n} I(\pi_n^{(i)},x^n,y^n) \ge \bar{I}(x^n,y^n) + \gamma  \right\} \nonumber \\
&\le& |{\cal X}^n| |{\cal Y}^n| \exp\{-(\gamma^2/4) m_n\}.
\label{eqn:doubly-exponential-bound}
\end{eqnarray}
Since $\exp\{-(\gamma^2/4) m_n\}$ converges to $0$
doubly exponentially, the right hand side of (\ref{eqn:doubly-exponential-bound})
is strictly smaller than $1$ if $n$ is sufficiently large, which implies that
there exists one realization of $\pi_1^{(1)},\ldots,\pi_n^{(m_n)}$ such that 
\begin{eqnarray}
\label{eq:before-average}
\frac{1}{m_n}\sum_{i=1}^{m_n} I(\pi_n^{(i)},x^n,y^n) \le \bar{I}(x^n,y^n) + \gamma
\end{eqnarray}
for every $(x^n,y^n)$.
Finally, by taking the average of both sides of (\ref{eq:before-average})
with respect to $(X^n,Y^n)$, we have the assertion of the lemma.
\end{proof}

Finally, by combining Lemma \ref{lemma:code-with-random-permutation} and 
Lemma \ref{lemma:after-de-randomize},  by taking
the constants to be sufficiently small and $n$ to be sufficiently large, 
we can show (\ref{eq:direct-1}). \qed

%%%%%%%%% Average Case Proof %%%%%%%%
\section{Proof of Theorem \ref{theorem:main-2}}
\label{section:proof-of-theorem:main-2}

\subsection{Proof of Converse Part}

We only prove (\ref{eq:rd-average-class-converse}) because (\ref{eq:rd-average-class-converse-2}) 
is obtained from (\ref{eq:rd-average-class-converse}) by
letting $E_1 = E_*$ and $W_1 = W_*$.

Assume that $R$ is achievable and fix $\lambda$, $E_1$, $E_2$, $W_1$, and $W_2$
such that $\lambda E_1 + (1-\lambda) E_2 \le E$ and
$W_j \in {\cal W}_j(E_j)$ for $j=1,2$. To prove (\ref{eq:rd-average-class-converse}), it is 
sufficient to show that there exists a pair $(D_1,D_2)$ such that
$\lambda D_1 + (1-\lambda) D_2 \le D$ and $R \ge R_{HB}(D_1,D_2|W_1,W_2)$.

To do this, we consider the compound channel 
\begin{eqnarray*}
W^n = \lambda W_1^{\times n} + (1-\lambda) W_2^{\times n}.
\end{eqnarray*}
Note that $\bm{W} = \{W^n \}_{n=1}^\infty \in {\cal W}_a(E)$ since
\begin{eqnarray*}
e_n(P_{X}^n, W^n) 
&=& \sum_{x^n,y^n} P_X^n(x^n) W^n(y^n|x^n) e_n(x^n,y^n) \\
&=& \lambda e_n(P_X^n,W_1^{\times n}) + (1-\lambda) e_n(P_X^n, W_2^{\times n}) \\
&\le& \lambda E_1 + (1-\lambda) E_2 \\
&\le& E.
\end{eqnarray*}
Hence, by the definition of the achievability of $R$, for arbitrary small $\varepsilon > 0$
and sufficiently large $n$, there exists a code $(\varphi_n,\psi_n)$ such that 
\begin{eqnarray*}
\frac{1}{n} \log |{\cal M}_n| \le R + \varepsilon
\end{eqnarray*}
and
\begin{eqnarray}
\label{eq:average-converse-distortion-condition}
\sum_{x^n,y^n} P_X^n(x^n) W^n(y^n|x^n) d_n(x^n, \psi_n(\varphi_n(x^n), y^n)) \le D + \varepsilon.
\end{eqnarray}
Note that (\ref{eq:average-converse-distortion-condition}) can be also written as
\begin{eqnarray}
D + \varepsilon 
&\ge& \lambda \sum_{x^n, y^n} P_X^n(x^n) W_1^{\times n}(y^n|x^n) d_n(x^n,\psi_n(\varphi_n(x^n),y^n)) \\
 &&~~ + (1-\lambda) \sum_{x^n,y^n} P_X^n(x^n) W_2^{\times n}(y^n|x^n) d_n(x^n,\psi_n(\varphi_n(x^n),y^n)).
\label{eq:average-converse-distortion-condition-2}
\end{eqnarray}
On the other hand, by using $(\varphi_n,\psi_n)$, we can construct a HB code 
$(\varphi_n^{HB}, \psi_n^{HB1}, \psi_n^{HB2})$ as
\begin{eqnarray*}
\varphi_n^{HB}(x^n) &=& \varphi_n(x^n)~~~x^n \in {\cal X}^n, \\
\psi_n^{HB1}(m,y_1^n) &=& \psi_n(m,y_1^n) ~~~m \in {\cal M}_n, y_1^n \in {\cal Y}^n, \\
\psi_n^{HB2}(m,y_2^n) &=& \psi_n(m,y_2^n) ~~~m \in {\cal M}_n, y_2^n \in {\cal Y}^n.
\end{eqnarray*}
Then, let $(D_1,D_2)$ be the pair of average distortion occurred by $(\varphi_n^{HB}, \psi_n^{HB1}, \psi_n^{HB2})$, i.e.,
\begin{eqnarray}
D_j := \sum_{x^n,y^n} P_X^n(x^n) W_j^{\times n}(y_j^n|x^n) d_n(x^n, \psi_n^{HBj}(\varphi_n^{HB}(x^n),y_j^n)) ~~~j=1,2.
\label{eq:average-converse-distortion-condition-3}
\end{eqnarray}
By the definition of $R_{HB}(D_1,D_2|W_1,W_2)$ and the construction of the code, we have
\begin{eqnarray*}
R + \varepsilon \ge R_{HB}(D_1,D_2|W_1,W_2).
\end{eqnarray*}
Further, (\ref{eq:average-converse-distortion-condition-2}) and (\ref{eq:average-converse-distortion-condition-3}) indicate
\begin{eqnarray*}
D + \varepsilon \ge \lambda D_1 + (1-\lambda) D_2.
\end{eqnarray*}
Since we can choose $\varepsilon$ arbitrary small, we have
$R \ge R_{HB}(D_1,D_2|W_1,W_2)$ and
$\lambda D_1 + (1-\lambda) D_2 \le D$. \qed

\subsection{Proof of Direct Part}

As the direct part proof of Theorem \ref{theorem:main}, we prove 
(\ref{eq:rd-average-class}) in three steps. First, we construct 
a code for i.i.d. channel. Then, it is used to permutation invariant channels
by using random permutation.
Then, the size of the randomness is reduced by the de-randomization technique.

\subsubsection{Code for i.i.d. channel}

The goal of this section is to show the following lemma.
\begin{lemma}
\label{lemma:average-case-iid}
For arbitrarily fixed $V \in {\cal V}(E,D)$ and $\delta > 0$,
there exists $\mu > 0$ and 
a code $\varphi_n^\prime:{\cal X}^n \to {\cal U}^n$ such that
\begin{eqnarray}
\label{eq:no-side-info-universal-code-rate}
\frac{1}{n}\log |\varphi_n^\prime| \le I(P_X,V) + 2 \delta
\end{eqnarray}
and
\begin{eqnarray*}
\Pr\left\{ (\hat{U}^n,X^n,Y^n) \notin T_{P_XV\bar{W},\delta} \right\} \le 2^{-\mu n}
\end{eqnarray*}
for every $\bar{W} \in {\cal P}_n({\cal Y}|{\cal X})$ 
provided that $n$ is sufficiently large, where $|\varphi_n^\prime|$ is the cardinality
of the image of $\varphi_n^\prime$, 
$\hat{U}^n = \varphi_n^\prime(X^n)$ and
$T_{P_XV\bar{W},\delta}$ is the set of $P_{UXY}$-typical sequences with
respect to $P_{UXY}(u,x,y) = P_X(x) V(u|x) \bar{W}(y|x)$.
\end{lemma}
\begin{proof}
We construct a code in a similar manner as Section \ref{subsection-iid}.
We use two kinds of bin codings $f_n:{\cal U}^n \to {\cal S}_n$ and
$g_n:{\cal U}^n \to {\cal L}_n$. We set 
\begin{eqnarray*}
R_f &=& H(U|X) - \delta, \\
R_g &=& I(P_X,V) + 2 \delta.
\end{eqnarray*}
Let $|{\cal S}_n| = \lfloor 2^{n R_f} \rfloor$ and $|{\cal L}_n| = \lceil 2^{n R_g} \rceil$.

Since
\begin{eqnarray*}
R_f + R_g = H(U) + \delta,
\end{eqnarray*}
there exists a decoder $\kappa_n:{\cal S}_n \times {\cal L}_n \to {\cal U}^n$ and $\mu_1 > 0$ such
that 
\begin{eqnarray}
\label{eq:error-bound-source-coding}
\mathbb{E}_{F_n G_n}\left[ P_{err}(F_n,G_n) \right] \le 2^{- \mu_1 n}
\end{eqnarray}
for sufficiently large $n$, where $P_{err}(F_n,G_n)$ is the error probability
of the source coding when the bin codings $(F_n,G_n)$ are used.
Furthermore, since $R_f = H(U|X) - \delta$, $S_n = F_n(U^n)$ is close to the 
uniform random variable that is independent of $X^n$ (Lemma \ref{lemma:privacy-amplification}).

We construct a code as follows. Let 
\begin{eqnarray*}
P_{\hat{U}^n|S_n L_n}(u^n|s_n,\ell_n) = \bol{1}[u^n = \kappa_n(s_n,\ell_n)]
\end{eqnarray*}
be the distribution describing the decoder. Let
\begin{eqnarray*}
\lefteqn{ P_{S_n L_n U^n X^n Y^n \hat{U}^n}(s_n,\ell_n,u^n,x^n,y^n,\hat{u}^n) } \\
&=& P_{S_n X^n}(s_n,x^n) P_{U^n|S_n X^n}(u^n|s_n,x^n) P_{L_n|U^n}(\ell_n|u^n)
 P_{Y^n|X^n}(y^n|x^n) P_{\hat{U}^n|S_n L_n}(\hat{u}^n|s_n,\ell_n)
\end{eqnarray*}
and 
\begin{eqnarray*}
\lefteqn{ \hat{P}_{\bar{S}_n L_n U^n X^n Y^n \hat{U}^n}(s_n,\ell_n,u^n,x^n,y^n,\hat{u}^n) } \\
&=& P_{\bar{S}_n}(s_n) P_{X^n}(x^n) P_{U^n|S_n X^n}(u^n|s_n,x^n) P_{L_n|U^n}(\ell_n|u^n)
 P_{Y^n|X^n}(y^n|x^n) P_{\hat{U}^n|S_n L_n}(\hat{u}^n|s_n,\ell_n).
\end{eqnarray*}
Note that $P_{U^n|S_n X^n}$ is a randomized quantizer.
From Lemma \ref{lemma:privacy-amplification} and the fact that the variational distance does not increase
by data processing and marginalization, we have
\begin{eqnarray*}
\mathbb{E}_{F_n G_n}\left[ \| \hat{P}_{\bar{S}_n U^n X^n Y^n \hat{U}^n} - P_{S_n U^n X^n Y^n \hat{U}^n} \| \right] \le 2^{- \mu_2 n}
\end{eqnarray*}
for some $\mu_2 > 0$.
By Lemma \ref{lemma:typicality} and (\ref{eq:error-bound-source-coding}), we have
\begin{eqnarray*}
\lefteqn{ \mathbb{E}_{F_n G_n}\left[ P_{S_n U^n X^n Y^n \hat{U}^n}(\{ (\hat{u}^n,x^n,y^n) \notin T_{P_XVW,\delta} \}) \right] } \\
&\le& \mathbb{E}_{F_n G_n}\left[ P_{S_n U^n X^n Y^n \hat{U}^n}(\{ (u^n,x^n,y^n) \notin T_{P_XVW,\delta}  \mbox{ or } u^n \neq \hat{u}^n \}) \right] \\
&\le& 2^{- \mu_3 n}
\end{eqnarray*}
for some $\mu_3 > 0$. Since 
\begin{eqnarray*}
\lefteqn{ \hat{P}_{\bar{S}_n U^n X^n Y^n \hat{U}^n}({\cal A}) - P_{S_n U^n X^n Y^n \hat{U}^n}({\cal A}) } \\
&\le& \| \hat{P}_{\bar{S}_n U^n X^n Y^n \hat{U}^n} - P_{S_n U^n X^n Y^n \hat{U}^n} \|
\end{eqnarray*}
for any set ${\cal A}$, we have
\begin{eqnarray*}
\mathbb{E}_{F_n G_n}\left[ \hat{P}_{\bar{S}_n U^n X^n Y^n \hat{U}^n}(\{ (\hat{u}^n, x^n,y^n) \notin T_{P_XVW,\delta} \}) \right] \le 2^{-\mu_4 n}
\end{eqnarray*}
for some $\mu_4 > 0$.
Since the cardinality of ${\cal P}_n({\cal Y}|{\cal X})$ is bounded by $(n+1)^{|{\cal X}||{\cal Y}|}$, there
exists one realization $(f_n,g_n,s_n)$ of $(F_n,G_n,S_n)$ satisfying 
\begin{eqnarray*}
\hat{P}_{U^n X^n Y^n \hat{U}^n|\bar{S}_n}(\{ (\hat{u}^n, x^n,y^n) \notin T_{P_XVW,\delta} \}|s_n) \le (n+1)^{|{\cal X}||{\cal Y}|} 2^{- \mu_4 n}.
\end{eqnarray*}
Furthermore, let $K_n$ be a random variable that simulate the randomized quantizer $P_{U^n|S_n X^n}$.
Then, we can also eliminate this randomness in a similar manner.
Let $\tau_n:{\cal S}_n \times {\cal X}^n \to {\cal U}^n$ be the resulting deterministic quantizer.
Then, we set $\varphi_n^\prime(x^n) = \kappa_n(s_n, g_n(\tau_n(s_n,x^n)))$.
The image size of $\varphi_n^\prime$ obviously satisfies (\ref{eq:no-side-info-universal-code-rate}).
Thus, by taking $n$ sufficiently large, we have the assertion of the lemma.
\end{proof}

%%%
\subsubsection{Code for Permutation Invariant Channel}

\begin{lemma}
\label{lemma:average-case-permutation-invariant}
For $\bar{W} \in \bar{{\cal W}}_n(T_{X,\delta})$, let ${\cal A}_n(\bar{W}) \subset {\cal X}^n \times {\cal Y}^n$. 
Suppose that
\begin{eqnarray*}
P_X^n \cdot \bar{W}^{\times n}({\cal A}_n(\bar{W})^c) \le \bar{\varepsilon}
\end{eqnarray*}
for every $\bar{W} \in \bar{{\cal W}}_n(T_{X,\delta})$. Then, for any conditional type
$\bar{W} \in \bar{{\cal W}}_n(T_{X,\delta})$, we have
\begin{eqnarray*}
\lefteqn{ \sum_{x^n \in T_{X,\delta}} P_X^n(x^n) \bol{1}[\bar{W} \in {\cal P}_n({\cal Y}|P_{x^n}) ] } \\
&& \times \sum_{y^n \in T_{\bar{W}}(x^n)} \frac{1}{|T_{\bar{W}}(x^n)|} \bol{1}[(x^n,y^n) \in {\cal A}_n(\bar{W})^c] \\
&\le& (n+1)^{|{\cal X}| |{\cal Y}|} \bar{\varepsilon}.
\end{eqnarray*}
\end{lemma}
\begin{proof}
We prove this lemma in a similar manner as Lemma \ref{lemma:permutation-invariant-probability}.
For any $\bar{W} \in {\cal P}_n({\cal Y}|P_{x^n})$, note that (\ref{eq:type-channel-probability}) holds.
Then, for any $\bar{W} \in \bar{{\cal W}}_n(T_{X,\delta})$, we have
\begin{eqnarray*}
\bar{\varepsilon} 
&\ge& P_X^n \cdot \bar{W}^{\times n}({\cal A}_n(\bar{W})^c) \\
&\ge& \sum_{x^n \in T_{X,\delta}} P_X^n(x^n) \bol{1}[\bar{W} \in {\cal P}_n({\cal Y}|P_{x^n})]
  \sum_{y^n \in T_{\bar{W}}(x^n)} \bar{W}^{\times n}(T_{\bar{W}}(x^n)|x^n) \frac{1}{|T_{\bar{W}}(x^n)|} \bol{1}[(x^n,y^n) \in {\cal A}_n(\bar{W})^c] \\
&\ge& \frac{1}{(n+1)^{|{\cal X}||{\cal Y}|}} \sum_{x^n \in T_{X,\delta}} P_X^n(x^n) \bol{1}[\bar{W} \in {\cal P}_n({\cal Y}|P_{x^n})]
  \sum_{y^n \in T_{\bar{W}}(x^n)}  \frac{1}{|T_{\bar{W}}(x^n)|} \bol{1}[(x^n,y^n) \in {\cal A}_n(\bar{W})^c],
\end{eqnarray*}
which implies the statement of the lemma.
\end{proof}

%%%
\begin{lemma}
\label{lemma:average-case-distortion-evaluation}
For a given $V \in {\cal V}(E + 2 \delta e_{\max},D)$, suppose that there exists
$\varphi_n^\prime:{\cal X}^n \to {\cal U}^n$ such that
\begin{eqnarray*}
\Pr\{ (\hat{U}^n,X^n,Y^n) \notin T_{P_XV\bar{W},\delta} \} \le \bar{\varepsilon}
\end{eqnarray*}
for every $\bar{W} \in {\cal P}_n({\cal Y}|{\cal X})$, where $\hat{U}^n = \varphi_n^\prime(X^n)$
and $T_{P_XV\bar{W},\delta}$ is the set of all $P_{UXY}$-typical set for $P_{UXY}(u,x,y) = P_X(x) V(u|x) \bar{W}(y|x)$.
Then, we have
\begin{eqnarray*}
\lefteqn{ \mathbb{E}\left[ d_n(X^n, \hat{U}^n(Y^n)) \right] } \\
&\le& \left\{ P_X^n(T_{X,\delta}^c) + (n+1)^{2|{\cal X}||{\cal Y}|} \bar{\varepsilon} + \delta \right\} d_{\max} 
 	+  D 
\end{eqnarray*}
for every permutation invariant (not necessarily i.i.d.) $W^n$ such that
\begin{eqnarray*}
\mathbb{E}[e_n(X^n,Y^n) ] \le E
\end{eqnarray*}
provided that $n$ is sufficiently large.
\end{lemma}
\begin{proof}
From (\ref{eq:triple-distortion-relation}),
we first note that
\begin{eqnarray*}
d_n(x^n,u^n(y^n)) &\le& d(V,\bar{W}) + \delta d_{\max}
\end{eqnarray*}
for $(u^n,x^n,y^n) \in T_{P_XV\bar{W},\delta}$.
Then, by using Lemma \ref{lemma:average-case-permutation-invariant} for
\begin{eqnarray*}
{\cal A}_n(\bar{W}) = \{ (x^n,y^n) : (\varphi_n^\prime(x^n), x^n,y^n) \in T_{P_XV\bar{W},\delta} \},
\end{eqnarray*}
we have
\begin{eqnarray*}
\lefteqn{ \mathbb{E}\left[ d_n(X^n, \hat{U}^n(Y^n)) \right] }  \\
&\le& P_X^n(T_{X,\delta}^c) d_{\max} \\
&& + \sum_{\bar{W} \in \bar{{\cal W}}_n(T_{X,\delta})} \sum_{x^n \in T_{X,\delta}} P_X^n(x^n) \bol{1}[\bar{W} \in {\cal P}_n({\cal Y}|P_{x^n})] \\
&& \sum_{y^n \in T_{\bar{W}}(x^n)} W^n(T_{\bar{W}}(x^n)|x^n) \frac{1}{|T_{\bar{W}}(x^n)|} \bol{1}[ (x^n,y^n) \in {\cal A}_n^c(\bar{W})] d_{\max} \\
&& + \sum_{\bar{W} \in \bar{{\cal W}}_n(T_{X,\delta})} \sum_{x^n \in T_{X,\delta}} P_X^n(x^n) \bol{1}[\bar{W} \in {\cal P}_n({\cal Y}|P_{x^n})] \\
&& \sum_{y^n \in T_{\bar{W}}(x^n)} W^n(T_{\bar{W}}(x^n)|x^n) \frac{1}{|T_{\bar{W}}(x^n)|} \bol{1}[ (x^n,y^n) \in {\cal A}_n(\bar{W})] 
  \{ d(V,\bar{W}) + \delta d_{\max} \}\\
&\le& \left\{ P_X^n(T_{X,\delta}^c) + (n+1)^{2|{\cal X}||{\cal Y}|} \bar{\varepsilon} + \delta \right\} d_{\max} \\
&& + \sum_{x^n \in T_{X,\delta} }\sum_{\bar{W} \in {\cal P}_n({\cal Y}|P_{x^n})} P_X^n(x^n) W^n(T_{\bar{W}}(x^n)|x^n) d(V,\bar{W}),
\end{eqnarray*}
where we used Lemma \ref{lemma:average-case-permutation-invariant} to upper bound the second term in
the second inequality.
Now, we rewrite the last term as 
\begin{eqnarray*}
&& \hspace{-10mm} \sum_{x^n \in T_{X,\delta} }\sum_{\bar{W} \in {\cal P}_n({\cal Y}|P_{x^n})} P_X^n(x^n) W^n(T_{\bar{W}}(x^n)|x^n) d(V,\bar{W})  \\
&=& P_{X}^n(T_{X,\delta}) d(V,W_{mix}),
\end{eqnarray*}
where $W_{mix} \in {\cal P}({\cal Y}|{\cal X})$ is a channel defined by
\begin{eqnarray*}
W_{mix}(y|x) = \sum_{x^n \in T_{X,\delta}} \tilde{P}_X^n(x^n) \sum_{\bar{W} \in {\cal P}_n({\cal Y}|P_{x^n}) }W^n(T_{\bar{W}}(x^n)|x^n) \bar{W}(y|x)
\end{eqnarray*}
for 
\begin{eqnarray*}
\tilde{P}_{X}^n(x^n) = \frac{P_X^n(x^n)}{P_X^n(T_{X,\delta})}.
\end{eqnarray*}
From (\ref{eq:distortion-empirical-dependency}) and (\ref{eq:distortion-relation}), we have 
\begin{eqnarray*}
e_n(x^n,y^n) \ge  e(P_X,\bar{W}) - \delta e_{\max}
\end{eqnarray*}
for $x^n \in T_{X,\delta}$ and $y^n \in T_{\bar{W}}(x^n)$.
Thus, we have
\begin{eqnarray*}
E &\ge& \mathbb{E}\left[ e_n(X^n,Y^n) \right] \\
&=& \sum_{x^n,y^n} P_X^n(x^n) W^n(y^n|x^n) e_n(x^n,y^n) \\
&\ge& \sum_{x^n \in T_{X,\delta}} P_X^n(x^n) \sum_{\bar{W} \in {\cal P}_n({\cal Y}|P_{x^n})} W^n(T_{\bar{W}}(x^n)|x^n) e(P_{x^n},\bar{W}) \\
&\ge& \sum_{x^n T_{X,\delta}} P_X^n(x^n) \sum_{\bar{W} \in {\cal P}_n({\cal Y}|P_{x^n})} W^n(T_{\bar{W}}(x^n)|x^n) \{e(P_X,\bar{W}) - \delta e_{\max} \} \\
&=& P_X^n(T_{X,\delta}) \{ e(P_X,W_{mix}) - \delta e_{\max} \}.   
\end{eqnarray*}
Thus, we have $W_{mix} \in {\cal W}_1(E + 2 \delta e_{\max})$ provided that $n$ is sufficiently large.
Since $V \in {\cal V}(E + 2 \delta e_{\max}, D)$,  we have $d(V, W_{mix}) \le D$.
This completes the proof.
\end{proof}

%%%
By combining Lemma \ref{lemma:average-case-iid} and Lemma \ref{lemma:average-case-distortion-evaluation}, 
we have the following.
\begin{lemma}
\label{lemma:average-case-randomized-permutation}
For any $V \in {\cal V}(E + 2 \delta e_{\max}, D)$,
$\delta > 0$, and $\varepsilon > 0$, there exists $\varphi_n^\prime:{\cal X}^n \to {\cal U}^n$ such that
\begin{eqnarray*}
\frac{1}{n} \log |\varphi_n^\prime| \le  I(P_X,V) + \delta
\end{eqnarray*} 
and
\begin{eqnarray}
\label{eq:average-case-permutation-averaged-distortion}
\mathbb{E}_{\pi_n}\left[  \mathbb{E}\left[ d_n(\pi_n(X^n), \hat{U}^n(\pi_n(Y^n))) \right] \right] \le D + \varepsilon
\end{eqnarray}
for every $\bol{W} \in {\cal W}_m(E)$ provided that $n$ is sufficiently large,
where $\hat{U}^n = \varphi_n^\prime(\pi_n(X^n))$.
\end{lemma}

%%%
\subsubsection{De-Randomization}

Now we reduce the size of random permutation by using the
de-randomization technique.
\begin{lemma}
\label{lemma:average-case-de-randomized}
Suppose that $\varphi_n^\prime$ satisfies (\ref{eq:average-case-permutation-averaged-distortion}).
Then, for arbitrary $\delta_2,\gamma > 0$, there exists $m_n = 2^{\delta_2 n}$ permutations
$\{ \pi_n^{(1)},\ldots,\pi_n^{(m_n)} \}$ such that
\begin{eqnarray*}
\frac{1}{m_n}\sum_{i=1}^{m_n} \mathbb{E}\left[ d_n(\pi_n(X^n), \hat{U}^n(\pi_n(Y^n))) \right] \le D + \varepsilon + \gamma
\end{eqnarray*}
provided that $n$ is sufficiently large.
\end{lemma}
\begin{proof}
For a permutation $\pi_n$ and $(x^n,y^n) \in {\cal X}^n \times {\cal Y}^n$, we denote
\begin{eqnarray*}
J(\pi_n,x^n,y^n) = d_n(\pi_n(x^n),\hat{u}^n(\pi_n(y^n))),
\end{eqnarray*}
where $\hat{u}^n = \varphi_n^\prime(\pi_n(x^n))$. Let 
$\pi_n^{(1)},\ldots,\pi_n^{(m_n)}$ be randomly generated permutations, and let
$\bar{J}(x^n,y^n) = \mathbb{E}_{\pi_n}[J(\pi_n,x^n,y^n)]$. Then, by using Lemma \ref{lemma:bernstein-trick}
for $A_i = J(\pi_n^{(i)}, x^n,y^n)$, $b= d_{\max}$, and $\alpha = \frac{\gamma}{2 d_{\max}^2}$, we have
\begin{eqnarray*}
&& \hspace{-10mm} \Pr\left\{ \frac{1}{m_n} \sum_{i=1}^{m_n} J(\pi_n^{(i)}, x^n,y^n) \ge \bar{J}(x^n,y^n) + \gamma \right\} \\
&\le& \exp\{ - (\gamma^2 / 4 d_{\max}^2) m_n \}. 
\end{eqnarray*}
Furthermore, by using the union bound, we have
\begin{eqnarray}
&& \hspace{-10mm} \Pr\left\{ \exists (x^n,y^n)~\frac{1}{m_n} \sum_{i=1}^{m_n} J(\pi_n^{(i)}, x^n,y^n) \ge \bar{J}(x^n,y^n) + \gamma \right\} \nonumber \\
&\le& |{\cal X}^n||{\cal Y}^n| \exp\{ - (\gamma^2 / 4 d_{\max}^2) m_n \}.
\label{eq:average-case-doubly-exponential-term} 
\end{eqnarray}
Since $\exp\{ - (\gamma^2 / 4 d_{\max}^2) m_n \}$ converges to $0$ doubly exponentially, the
righthand side of (\ref{eq:average-case-doubly-exponential-term}) is strictly smaller than $1$ if
$n$ is sufficiently large, which implies that there exists one realization of $\pi_n^{(1)}, \ldots, \pi_n^{(m_n)}$ such that
\begin{eqnarray}
\label{eq:average-case-after-de-randomization}
\frac{1}{m_n} \sum_{i=1}^{m_n} J(\pi_n^{(i)},x^n,y^n) \le \bar{J}(x^n,y^n) + \gamma
\end{eqnarray}
for every $(x^n,y^n)$. Finally, by taking the average over both sides of 
(\ref{eq:average-case-after-de-randomization}) with respect to $(X^n,Y^n)$, we have
the assertion of the lemma.
\end{proof}

Finally, by combining Lemma \ref{lemma:average-case-randomized-permutation} and 
Lemma \ref{lemma:average-case-de-randomized}, and by taking
the constants to be sufficiently small and $n$ to be sufficiently large, 
we can show that the righthand side of (\ref{eq:rd-average-class}) is achievable.
\qed

%%% Discussion %%%%%
\section{Conclusion}

In this paper, we introduced the novel rate-distortion
functions for the Wyner-Ziv problem, which are defined 
as the minimum rates required for the universal coding 
for the distortion constrained general channel classes.
Then, we derived the upper bounds and lower bounds 
on the rate-distortion functions. The complete solution
for the rate-distortion functions is remained open.
Parts of difficulties are related to the Heegard-Berger problem,
which is also a long-standing open problem.

%%% Ack %%%%%%%%%
%\section*{Acknowledgment}

%This research is partly supported by
%Grand-in-Aid for Young Scientists(B):2376033700
%and...

%%%% Appendix %%%%%%
\appendix

%%%% Relation between Max and Average %%%%%
\subsection{Proof of Proposition \ref{proposition:relation-max-ave}}
\label{appendix:relation-max-ave}

From the definition, for any $\nu > 0$ there exists an average-achievable rate $R$ such that
$R \le R_a(D|E) + \nu$. For any $\epsilon > 0$ and $\bm{W} \in {\cal W}_m(E-\epsilon)$, we have
\begin{eqnarray}
\mathbb{E}\left[ e_n(P_{X^n},W^n) \right] &\le& E - \epsilon + e_{\max} \Pr\left\{ e_n(X^n,Y^n) > E \right\} \\
&\le& E
\label{eq:remation-max-ave-proof}
\end{eqnarray}
provided that $n \ge n_0(\epsilon/e_{\max})$. 

Let $\tilde{\bm{W}} = \{ \tilde{W}^n \}_{n=1}^\infty$ be a sequence of channels such that
$\tilde{W}^n = W^n$ for $n \ge n_0(\epsilon/e_{\max})$ and $\tilde{W}^n$ for $1 \le n < n_0(\epsilon/e_{\max})$
are chosen appropriately so that $\mathbb{E}\left[ e_n(P_{X^n},\tilde{W}^n) \right] \le E$. Then from \eqref{eq:remation-max-ave-proof},
we have $\tilde{\bm{W}} \in {\cal W}_a(E)$.
Then, since $R$
is average achievable, for any $\varepsilon > 0$ there exists a code such that \eqref{eq:rate-requirement} and \eqref{eq:distortion-requirement}
are satisfied for $\tilde{W}^n$ and $n \ge n_1(\varepsilon)$.
This also implies that the code also satisfies \eqref{eq:rate-requirement} and \eqref{eq:distortion-requirement} for
$W^n$ and $n \ge \max[n_0(\epsilon/e_{\max}), n_1(\varepsilon)]$. Since $\bm{W} \in {\cal W}_m(E-\epsilon)$ is arbitrary,
$R$ is also maximum-achievable. Since $\nu > 0$ is arbitrary, we have
\begin{eqnarray*}
R_a(D|E) \ge R_m(D|E-\epsilon).
\end{eqnarray*}
Thus, by taking the limit $\epsilon \to 0$, we have the assertion of the proposition. \qed

%%%% Continuity of R_WZ(D|W) %%%%%%%%%%
\subsection{Continuity of $R_{WZ}(D|W)$}
\label{appendix:continuity-of-RWZ}

\begin{lemma}
For $D > 0$, $R_{WZ}(D|W)$ is a continuous function with respect to $W$.
\end{lemma}

\begin{proof}
For two channels $W_1,W_2$, we consider the distance given by
\begin{eqnarray*}
\Delta(W_1,W_2) := \sum_{x,y} |P_X(x) W_1(y|x) - P_X(x)W_2(y|x) |.
\end{eqnarray*}
Since the Euclidian distance $\|W_1-W_2\|_2$ converging to $0$ implies $\Delta(W_1,W_2)$ converging to $0$,
it suffice to show the continuity of $R_{WZ}(D|W)$ with respect to the topology given by $\Delta(\cdot,\cdot)$.

By a slight abuse of notation, we also introduce 
\begin{eqnarray*}
\Delta(V_1,V_2) := \sum_{u,x}|P_X(x)V_1(u|x) - P_X(x)V_2(u|x)|
\end{eqnarray*}
for two test channel $V_1,V_2$.
From the definition of the variational distance, we can find that 
\begin{eqnarray*}
\Delta(W_1,W_2) = \|P_X V W_1 - P_X V W_2 \|
\end{eqnarray*}
for any fixed test channel $V$ and 
\begin{eqnarray*}
\Delta(V_1,V_2) = \| P_X V_1 W - P_X V_2 W \|
\end{eqnarray*}
for any fixed channel $W$, where $P_X V W$ is the joint distribution
given by $P_X(x)V(u|x)W(y|x)$. Furthermore, from Fannes' inequality \cite[Lemma 2.7]{csiszar-korner:11}, 
there exists a function $\nu(\delta)$ such that $\nu(\delta) \to 0$ as $\delta \to 0$ and
\begin{eqnarray} \label{eq:app-cont-w}
|\phi(V,W_1) - \phi(V,W_2)| \le \nu(\Delta(W_1,W_2))
\end{eqnarray} 
for fixed $V$ and
\begin{eqnarray} \label{eq:app-cont-v}
|\phi(V_1,W) - \phi(V_2,W)| \le \nu(\Delta(V_1,V_2))
\end{eqnarray}
for fixed $W$.

For two channels $W_1,W_2$, let 
$\epsilon :=  d_{\max} \Delta(W_1,W_2)$. Let $V_i^*  \in {\cal V}(W_i,D+\epsilon)$ be a test channel such that
\begin{eqnarray*}
\phi(V_i^*,W_i) = R_{WZ}(D+\epsilon|W_i).
\end{eqnarray*}
Let $\tilde{V}_i$ be a test channel such that $\tilde{V}_i \in {\cal V}(W_i,0)$. 
We set 
\begin{eqnarray*}
V_i^\dagger := \frac{D}{D + \epsilon} V_i^* + \frac{\epsilon}{D+\epsilon} \tilde{V}_i. 
\end{eqnarray*}
Then, we have
\begin{eqnarray}
\Delta(V_i^*,V_i^\dagger) &\le& \frac{\epsilon}{D+\epsilon} \| P_X V_i^* - P_X \tilde{V}_i  \| \\
&\le&  \frac{2 \epsilon }{D+\epsilon} 
\label{eq:dist-v-v-dagger}
\end{eqnarray}
and
\begin{eqnarray*}
d(V_i^\dagger,W_i) &=& \frac{D}{D+\epsilon} d(V_i^*,W_i) + \frac{\epsilon}{D+\epsilon} d(\tilde{V}_i,W_i) \\
&\le& D.
\end{eqnarray*}
Furthermore, let $V_i^{\ddagger} \in {\cal V}(W_i,D)$ be such that
\begin{eqnarray*}
\phi(V_i^{\ddagger},W_i) = R_{WZ}(D|W_i).
\end{eqnarray*}
Note that $V_1^\ddagger \in {\cal V}(W_2,D+\epsilon)$ and $V_2^\ddagger \in {\cal V}(W_1,D+\epsilon)$.

By using above notations, we have
\begin{eqnarray*}
R_{WZ}(D|W_1) 
&=& \phi(V_1^\ddagger,W_1) \\
&\stackrel{(\rom{a})}{\ge}& \phi(V_1^\ddagger,W_2) - \nu(\Delta(W_1,W_2)) \\
&\stackrel{(\rom{b})}{\ge}& \phi(V_2^*,W_2) - \nu(\Delta(W_1,W_2)) \\
&\stackrel{(\rom{c})}{\ge}& \phi(V_2^\dagger,W_2) - \nu(\Delta(W_1,W_2)) - \nu\left(2\epsilon/ (D+\epsilon)\right) \\
&\stackrel{(\rom{d})}{\ge}& \phi(V_2^\ddagger,W_2) - \nu(\Delta(W_1,W_2)) - \nu\left(2\epsilon/ (D+\epsilon)\right) \\
&=& R_{WZ}(D|W_2)   - \nu(\Delta(W_1,W_2)) - \nu\left(2\epsilon/ (D+\epsilon)\right),
\end{eqnarray*}
where $(\rom{a})$ follows from \eqref{eq:app-cont-w}, $(\rom{b})$ follows from $V_1^\ddagger \in {\cal V}(W_2,D+\epsilon)$ and the fact that 
$V_2^*$ minimizes $\phi(V,W_2)$ under $V \in {\cal V}(W_2,D+\epsilon)$,
$(\rom{c})$ follows from \eqref{eq:dist-v-v-dagger} and \eqref{eq:app-cont-v},
and $(\rom{d})$ follows from $V_2^\dagger \in {\cal V}(W_2,D)$ and the fact that $V_2^\ddagger$ minimizes
$\phi(V,W_2)$ under $V \in {\cal V}(W_2,D)$.
Similarly, we can prove the inequality in which $W_1$ and $W_2$ are interchanged. Thus we have
\begin{eqnarray*}
| R_{WZ}(D|W_1) - R_{WZ}(D|W_2) | \le \nu(\Delta(W_1,W_2)) + \nu\left(2\epsilon/ (D+\epsilon)\right).
\end{eqnarray*}
Since $\epsilon \to 0$ as $\Delta(W_1,W_2) \to 0$, we have proved the continuity of $R_{WZ}(D|W)$.
\end{proof}

%%%% Type and Related Topics %%%%%%%%%%
\subsection{Miscellaneous Facts on Types and Typicality}
\label{Appendix:type}

In this section, we introduce some notations and known 
facts on the type method \cite{csiszar-korner:11}.

The type of a sequence $x^n$ and the joint type of 
$(x^n,y^n)$ are denoted by $P_{x^n}$ and $P_{x^n y^n}$ respectively.
The set of all types and joint types are denoted by 
${\cal P}_n({\cal X})$ and ${\cal P}_n({\cal X} \times {\cal Y})$.
For type $P$, the set of all sequence such that $P_{x^n} = P$ is
denoted by $T_P$. We use a similar notation for joint types. 
The set of all conditional types is denoted by ${\cal P}_n({\cal Y}|{\cal X})$,
and the set of $W$-shell for given $x^n$ is denoted by $T_W(x^n)$.
For type $P \in {\cal P}_n({\cal X})$, the set of all conditional types
such that $T_W(x^n)$ is not empty is denoted by ${\cal P}_n({\cal Y}|P)$.
It is well known that
\begin{eqnarray*}
|{\cal P}_n({\cal X})| &\le& (n+1)^{|{\cal X}|}, \\
|{\cal P}_n({\cal X} \times {\cal Y})| &\le& (n+1)^{|{\cal X}||{\cal Y}|}, \\
|{\cal P}_n({\cal Y}|{\cal X})| &\le& (n+1)^{|{\cal X}| |{\cal Y}|},
\end{eqnarray*}
and these inequalities are extensively used in the paper.

For $P_X \in {\cal P}({\cal X})$, a sequence $x^n$ is called $P_X$-typical
sequence with constant $\delta$ if
\begin{eqnarray*}
|P_{x^n}(a) - P_X(a)| \le \delta ~\forall a \in {\cal X}
\end{eqnarray*}
and no $a \in {\cal X}$ with $P_X(a) = 0$ occurs in $x^n$.
The set of all typical sequence is denoted by $T_{X,\delta}$.
The set of all types $P \in {\cal P}_n({\cal X})$ such that $T_P \subset T_{X,\delta}$
is denoted by ${\cal P}_{X,\delta,n}$.
For joint probability distribution, joint typical sequence 
and the set of all joint typical sequences are defined in a similar manner.
It is well known that the set of all non-typical sequences occur
with exponential small probability.
Especially for our purpose, we need a bound such that the convergence
is uniform with respect to $P_X$.
\begin{lemma}
\label{lemma:typicality}
For any $P_X \in {\cal P}({\cal X})$, we have
\begin{eqnarray*}
P_X^n(T_{X,\delta}^c) \le 2 |{\cal X}| 2^{- n \frac{ 2 \delta^2 }{5 \ln 2}}.
\end{eqnarray*}
\end{lemma}
\begin{proof}
For each $a \in {\cal X}$ such that $P_X(a) > 0$,
by noting that the variance of $\bol{1}[X_i = a] - P_X(a)$ is bounded by $\frac{1}{2}$ and
$|\bol{1}[X_i = a] - P_X(a)| \le 1$ with probability one, and 
by using the Bernstein inequality, we have
\begin{eqnarray*}
\Pr\{ | P_{X^n}(a) - P_X(a) | \ge \delta \} \le 2 \cdot 2^{- n \frac{ 2 \delta^2 }{5 \ln 2}}
\end{eqnarray*}
for any $0 < \delta \le1$. Thus, by using the union bound with respect to $a \in {\cal X}$, we have the assertion.
\end{proof}

Since the distortion is additive, the distortion between $x^n$ and $y^n$ 
only depends on their joint type, and thus we have
\begin{eqnarray}
\label{eq:distortion-empirical-dependency}
e_n(x^n,y^n) = e(P, W)
\end{eqnarray}
if $x^n \in T_P$ and $y^n \in T_W(x^n)$.
From the definition of ${\cal P}_{X,\delta,n}$, we have
\begin{eqnarray}
\label{eq:distortion-relation}
| e(P, W) - e(P_X,W) |\le  \delta e_{\max}
\end{eqnarray}
for any $P \in {\cal P}_{X,\delta,n}$ and $W \in {\cal P}({\cal Y}|{\cal X})$.

For $P \in {\cal P}_n({\cal X})$, let
\begin{eqnarray*}
\bar{{\cal W}}_n(P, E) := {\cal W}_1(P, E) \cap {\cal P}_n({\cal Y}|P).
\end{eqnarray*}
Then, for $P \in {\cal P}_{X,\delta,n}$, (\ref{eq:distortion-relation}) implies 
\begin{eqnarray}
\label{eq:typicality-implication}
W \in \bar{{\cal W}}_n(P, E) \Longrightarrow W \in {\cal W}_1(E + \delta e_{\max}).
\end{eqnarray}
We also use the notation
\begin{eqnarray*}
\bar{{\cal W}}_n(T_{X,\delta},E) &:=& \bigcup_{P \in {\cal P}_{X,\delta,n}} \bar{{\cal W}}_n(P,E), \\
\bar{{\cal W}}_n(T_{X,\delta}) &:=& \bigcup_{P  \in {\cal P}_{X,\delta,n}} {\cal P}_n({\cal Y}|P).
\end{eqnarray*}

For $(V,W) \in {\cal P}({\cal U}|{\cal X}) \times {\cal P}({\cal Y}|{\cal X})$, let 
$P_{UXY}(u,x,y) = P_X(x) V(y|x) W(y|x)$. For $(u^n,x^n,y^n) \in T_{UXY,\delta}$, 
by the same reason as (\ref{eq:distortion-empirical-dependency}) and (\ref{eq:distortion-relation}), we have
\begin{eqnarray}
\label{eq:triple-distortion-relation}
|d_n(x^n,u^n(y^n)) - d(V,W)| \le \delta d_{\max}.
\end{eqnarray}

%%%
\subsection{Privacy Amplification Lemma}
\label{proof-of-lemma:privacy-amplification}

\begin{lemma}
\label{lemma:privacy-amplification}
Let $F_n$ be the random binning from ${\cal U}^n$ to ${\cal S}_n$ such that
$|{\cal S}_n| = \lfloor 2^{n R_f} \rfloor$, where $R_f = H(U|X) - \delta$.
Then, there exists $\mu_2 > 0$ such that
\begin{eqnarray*}
\mathbb{E}_{F_n} \left[ \| P_{S_n X^n} - P_{\bar{S}_n} \times P_{X^n} \| \right] \le 2^{- \mu_2 n},
\end{eqnarray*}
where $P_{\bar{S}_n}$ is the uniform distribution on ${\cal S}_n$.
\end{lemma}
\begin{proof}
The lemma is a straightforward consequence of \cite[(51)]{hayashi:10b},
which states that
\begin{eqnarray}
\mathbb{E}\left[ \|P_{S_n X^n} - P_{\bar{S}_n} \times P_{X^n} \| \right] \le
 3 |{\cal S}_n|^{\theta} 2^{n \tau(\theta|P_{UX})}
\label{eq:pa-hayashi}
\end{eqnarray}
for $0 \le \theta \le \frac{1}{2}$, where 
\begin{eqnarray*}
\tau(\theta|P_{UX}) = \log \sum_x P_X(x) \left( \sum_{u} P_{U|X}(u|x)^{\frac{1}{1-\theta}} \right)^{1-\theta}.
\end{eqnarray*}
Since $\left. \frac{d \tau(\theta|P_{UX})}{d \theta} \right|_{\theta = 0} = - H(U|X)$, there exists $\theta_0 > 0$
such that
\begin{eqnarray*}
\frac{\tau(\theta_0|P_{UX})}{\theta_0} \le - H(U|X) + \frac{\delta}{2}.
\end{eqnarray*}
Thus, we have
\begin{eqnarray}
\frac{\theta_0}{n} \log |{\cal S}_n| + \tau(\theta_0|P_{UX}) 
	&\le& R_f - H(U|X) + \frac{\delta}{2} \nonumber \\
	&=& - \frac{\delta}{2}.  \label{proof-lemma-pa-exponent-evaluation}
\end{eqnarray}
Combining (\ref{eq:pa-hayashi}) 
and (\ref{proof-lemma-pa-exponent-evaluation}),
we have the assertion of the lemma.
\end{proof}

%%%%
\subsection{Bernstein's Trick}

\begin{lemma}[\cite{ahlswede:80}]
\label{lemma:bernstein-trick}
Let $A_1, \ldots,A_m$ be a sequence of
discrete independent random variables that take
values in $[-b,b]$. Then, for $0 < \alpha \le \min[1, \frac{b^n}{2} e^{-2b}]$, we have
\begin{eqnarray*}
\Pr\left\{ \frac{1}{m} \sum_{i=1}^m (A_i - \mathbb{E}[A_i]) \ge \gamma \right\}
\le \exp\left\{ (-\alpha \gamma + \alpha^2 b^2) m \right\}.
\end{eqnarray*}
\end{lemma}

%%%
%\subsection{Proof of Lemma \ref{lemma:approximation}}
%\label{proof-of-lemma:approximation}

\bibliographystyle{../09-04-17-bibtex/IEEEtran}
\bibliography{../09-04-17-bibtex/reference.bib}

\begin{thebibliography}{10}
\providecommand{\url}[1]{#1}
\csname url@samestyle\endcsname
\providecommand{\newblock}{\relax}
\providecommand{\bibinfo}[2]{#2}
\providecommand{\BIBentrySTDinterwordspacing}{\spaceskip=0pt\relax}
\providecommand{\BIBentryALTinterwordstretchfactor}{4}
\providecommand{\BIBentryALTinterwordspacing}{\spaceskip=\fontdimen2\font plus
\BIBentryALTinterwordstretchfactor\fontdimen3\font minus
  \fontdimen4\font\relax}
\providecommand{\BIBforeignlanguage}[2]{{%
\expandafter\ifx\csname l@#1\endcsname\relax
\typeout{** WARNING: IEEEtran.bst: No hyphenation pattern has been}%
\typeout{** loaded for the language `#1'. Using the pattern for}%
\typeout{** the default language instead.}%
\else
\language=\csname l@#1\endcsname
\fi
#2}}
\providecommand{\BIBdecl}{\relax}
\BIBdecl

\bibitem{wyner:76}
A.~D. Wyner and J.~Ziv, ``The rate-distortion function for source coding with
  side information at the decoder,'' \emph{IEEE Trans. Inform. Theory},
  vol.~22, no.~1, pp. 1--10, January 1976.

\bibitem{slepian:73}
D.~Slepian and J.~K. Wolf, ``Noiseless coding of correlated information
  sources,'' \emph{IEEE Trans. Inform. Theory}, vol.~19, no.~4, pp. 471--480,
  July 1973.

\bibitem{csiszar:81}
I.~Csisz\'ar and K\"orner, ``Graph decomposition: {A} new key to coding
  theorems,'' \emph{IEEE Trans. Inform. Theory}, vol.~27, no.~1, pp. 5--12,
  January 1981.

\bibitem{csiszar:82}
I.~Csisz\'ar, ``Linear codes for sources and source networks: Error exponents,
  universal coding,'' \emph{IEEE Trans. Inform. Theory}, vol.~28, no.~4, pp.
  585--592, July 1982.

\bibitem{oohama:94}
Y.~Oohama and T.~S. Han, ``Universal coding for the {S}lepian-{W}olf data
  compression system and the strong converse theorem,'' \emph{IEEE Trans.
  Inform. Theory}, vol.~40, no.~6, pp. 1908--1919, November 1994.

\bibitem{oohama:96}
Y.~Oohama, ``Universal coding for correlated sources with linked encoders,''
  \emph{IEEE Trans. Inform. Theory}, vol.~42, no.~3, pp. 837--843, May 1996.

\bibitem{kimura:09}
A.~Kimura, T.~Uyematsu, S.~Kuzuoka, and S.~Watanabe, ``Universal source coding
  over generalized complementary delivery networks,'' \emph{IEEE Trans. Inform.
  Theory}, vol.~55, no.~3, pp. 1360--1373, March 2009.

\bibitem{merhav:06}
N.~Merhav and J.~Ziv, ``On the {W}yner-{Z}iv problem for individual
  sequences,'' \emph{IEEE Trans. Inform. Theory}, vol.~52, no.~3, pp. 867--873,
  March 2006.

\bibitem{jalali:10}
S.~Jalali, S.~Verdu, and T.~Weissman, ``A universal scheme for {W}yner-{Z}iv
  coding of discrete sources,'' \emph{IEEE Trans. Inform. Theory}, vol.~56,
  no.~4, pp. 1737--1750, April 2010.

\bibitem{reani:11}
A.~Reani and N.~Merhav, ``Efficient on-line scheme for encoding individual
  sequences with side information at the decoder,'' \emph{IEEE Trans. Inform.
  Theory}, vol.~57, no.~10, pp. 6860--6876, October 2011.

\bibitem{kuzuoka:10b}
S.~Kuzuoka, A.~Kimura, and T.~Uyematsu, ``Universal source coding for multiple
  decoders with side information,'' in \emph{IEEE International Symposium on
  Information Theory}, 2010, pp. 1--5.

\bibitem{heegard:85}
C.~Heegard and T.~Berger, ``Rate distortion when side information may be
  absent,'' \emph{IEEE Trans. Inform. Theory}, vol.~31, no.~6, pp. 727--734,
  November 1985.

\bibitem{steinberg:04}
Y.~Steinberg and N.~Merhav, ``On successive refinement for the {W}yner-{Z}iv
  problem,'' \emph{IEEE Trans. Inform. Theory}, vol.~50, no.~8, pp. 1636--1654,
  August 2004.

\bibitem{tian:07}
C.~Tian and S.~Diggavi, ``On multistage successive refinement for {W}yner-{Z}iv
  source coding with degraded side informations,'' \emph{IEEE Trans. Inform.
  Theory}, vol.~53, no.~8, pp. 2946--2960, August 2007.

\bibitem{tian:08}
------, ``Side-information scalable source coding,'' \emph{IEEE Trans. Inform.
  Theory}, vol.~54, no.~12, pp. 5591--508, December 2008.

\bibitem{timo:11b}
R.~Timo, T.~Chan, and A.~Grant, ``Rate distortion with side-information at many
  decoders,'' \emph{IEEE Trans. Inform. Theory}, vol.~57, no.~8, pp.
  5240--5257, August 2011.

\bibitem{watanabe:11}
S.~Watanabe, ``The rate-distortion function for product of two sources with
  side-information at decoders,'' in \emph{Proc. IEEE Int. Symp. Inf. Theory
  2011}, Saint Petersburg, Russia, 2011, pp. 2862--2866, arXiv:1105.2864.

\bibitem{equitz:91}
W.~H.~R. Equitz and T.~M. Cover, ``Successive refinement of information,''
  \emph{IEEE Trans. Inform. Theory}, vol.~37, no.~2, pp. 269--275, March 1991.

\bibitem{rimoldi:94}
B.~Rimoldi, ``Successive refinement of information: Characterization of the
  achievable rates,'' \emph{IEEE Trans. Inform. Theory}, vol.~40, no.~1, pp.
  253--259, January 1994.

\bibitem{moulin:03}
P.~Mouline and J.~A. O'{S}ullivan, ``Information-theoretic analysis of
  information hiding,'' \emph{IEEE Trans. Inform. Theory}, vol.~49, no.~3, pp.
  563--593, March 2003.

\bibitem{cohen:02}
A.~S. Cohen and A.~Lapidoth, ``The gaussian watermarking game,'' \emph{IEEE
  Trans. Inform. Theory}, vol.~48, no.~6, pp. 1639--1667, June 2002.

\bibitem{baruch:03}
A.~S.-Baruch and N.~Merhav, ``On the error exponent and capacity games of
  private watermarking systems,'' \emph{IEEE Trans. Inform. Theory}, vol.~49,
  no.~3, pp. 537--562, March 2003.

\bibitem{barron:03}
R.~J. Barron, B.~Chen, and G.~W. Wornell, ``The duality between information
  embedding and source coding with side information and some applications,''
  \emph{IEEE Trans. Inform. Theory}, vol.~49, no.~5, pp. 1159--1180, May 2003.

\bibitem{cover:02}
T.~M. Cover and M.~Chiang, ``Duality between channel capacity and rate
  distortion with two-sided state information,'' \emph{IEEE Trans. Inform.
  Theory}, vol.~48, no.~6, pp. 1629--1638, June 2002.

\bibitem{elgamal-kim-book}
A.~{El Gamal} and Y.-H. Kim, \emph{Network Information Theory}.\hskip 1em plus
  0.5em minus 0.4em\relax Cambridge, 2011.

\bibitem{willems:83}
F.~M.~J. Willems, ``Computation of the {W}yner-{Z}iv rate distortion
  function,'' in \emph{Eindhoven Univ. Tech. Rep. 83-E-140}, 1983.

\bibitem{bertsekas:03}
D.~P. Bertsekas, A.~Nedi\'c, and A.~E. Ozdaglar, \emph{Convex Analysis and
  Optimization}.\hskip 1em plus 0.5em minus 0.4em\relax Athena Scientific,
  2003.

\bibitem{ahlswede:80}
R.~Ahlswede, ``A method of coding and an application to arbitrary varying
  channel,'' \emph{J. Comb. Inform. Syst.}, vol.~5, no.~1, pp. 10--35, 1980.

\bibitem{ziv:84}
J.~Ziv, ``Fixed-rate encoding of individual sequences with side information,''
  \emph{IEEE Trans. Inform. Theory}, vol.~30, no.~2, pp. 348--352, March 1984.

\bibitem{yassaee:12}
M.~H. Yassaee, M.~R. Aref, and A.~Gohari, ``Achievability proof via output
  statistics of random binning,'' 2012, arXiv:1203.0730.

\bibitem{kelly:12}
B.~G. Kelly and A.~B. Wagner, ``Reliability in source coding with side
  information,'' \emph{IEEE Trans. Inform. Theory}, vol.~58, no.~8, pp.
  5086--5111, August 2012, arXiv:1109.0923.

\bibitem{csiszar-korner:11}
I.~Csisz\'ar and J.~K\"orner, \emph{Information Theory, Coding Theorems for
  Discrete Memoryless Systems}, 2nd~ed.\hskip 1em plus 0.5em minus 0.4em\relax
  Cambridge University Press, 2011.

\bibitem{hayashi:10b}
M.~Hayashi, ``Tight exponential evaluation for information theoretical secrecy
  based on $l_1$ distance,'' arXiv:1010.1358.

\end{thebibliography}

%%%%% Author Bib %%%%%%%%%%%%%%%%%

%\begin{IEEEbiography}{Shun Watanabe}
%received the B.E., 
%M.E., and Ph.D.\ degrees from Tokyo Institute of Technology
%in 2005, 2007, and 2009 respectively. He is currently
%an Assistant Professor in the Department of Information 
%Science and Intelligent Systems of  University of Tokushima.
%His current research interests are in the areas of
%information theory, quantum information theory,
%and quantum cryptography.
%\end{IEEEbiography}

\end{document}